\documentclass{article}
\usepackage{graphicx}
\usepackage{amsmath, amsthm, latexsym}
\usepackage{amssymb, amsfonts,mathrsfs,amsthm}
\usepackage[mathscr]{euscript}

\newtheorem{theorem}{Theorem}
\newtheorem{lemma}{Lemma}
\newtheorem{definition}{Definition}
\newtheorem{corollary}{Corollary}

\newtheorem{remark}{Remark}
\newtheorem{example}{Example}
\usepackage{color}
\usepackage{hyperref} 
\usepackage[font={small,it}]{caption}

\hypersetup{colorlinks=true}    

\newcommand{\der}[2]{\frac{d#1}{d#2}} 
\newcommand{\pder}[2]{\dfrac{\partial#1}{\partial#2}} 
\newcommand{\tpder}[2]{\tfrac{\partial#1}{\partial#2}} 

\begin{document}

\vspace*{3cm} \thispagestyle{empty}
\vspace{5mm}

\noindent \textbf{\Large Pre-big bang geometric extensions of\\ 
inflationary cosmologies}\\

\textbf{\normalsize  \textbf{\normalsize David Klein\footnote{Department of Mathematics and Interdisciplinary Research Institute for the Sciences, California State University, Northridge, Northridge, CA 91330-8313. Email: david.klein@csun.edu.} and Jake Reschke}\footnote{Department of Mathematics, University of California, Davis, Davis, CA 95616
 Email: jreschke@math.ucdavis.edu }}
\\

\vspace{4mm} \parbox{11cm}{\noindent{\small  Robertson-Walker spacetimes within a large class are geometrically extended to larger cosmologies that include spacetime points with zero and negative cosmological times.  In the extended cosmologies, the big bang is lightlike, and though singular, it inherits some geometric structure from the original spacetime. Spacelike geodesics are continuous across the cosmological time zero submanifold which is parameterized by the radius of Fermi space slices, i.e, by the proper distances along spacelike geodesics from a comoving observer to the big bang. The continuous extension of the metric, and the continuously differentiable extension of the leading Fermi metric coefficient $g_{\tau\tau}$ of the observer, restrict the geometry of spacetime points with pre-big bang cosmological time coordinates.  In our extensions the big bang is two dimensional in a certain sense, consistent with some findings in quantum gravity.}\\


}\\
\vspace{6cm}
\pagebreak

\setlength{\textwidth}{27pc}
\setlength{\textheight}{43pc}

\section{Introduction}

The big bang singularity in general relativistic cosmologies can be considered from a variety of perspectives.  In quantum gravity theories the singularity can be  eliminated, with the big bang preceded by a big crunch or arising through other scenarios \cite{Ash, loop,  string2, string, string3}. Some investigations suggest that dimensional reduction may be a fundamental feature of quantum gravity with the effective dimension of spacetime points at sufficiently small scales decreasing to $d=2$ \cite{Carlip, Stoica}.\\  

\noindent The geometry near the singularity and possible pre-big bang scenarios  have also received  attention from a classical perspective.  Penrose and other researchers have approached this through conformal geometric methods \cite{Penrose, Penrose2, Penrose3, prebang2, Stoica2}, and dynamics near and through the big bang have also been investigated \cite{Belbruno}.  \\

\noindent In this paper, using the framework of general relativity, we show that Robertson-Walker spacetimes with big bang singularities can be extended to larger cosmologies that include points which, in a natural way, may be assigned negative or zero cosmological times.  Our approach begins with the observation that cosmological time along spacelike geodesics orthogonal to the worldline of a comoving observer decreases monotonically, and the geodesics terminate in finite proper distance at the big bang (c.f. \cite{randles, klein13}). One should therefore be able to construct larger cosmologies by extending these geodesics further, while at the same time preserving some continuity and differentiability properties of the metric tensor.  In the language of coordinates, Fermi charts of comoving observers are extended beyond their maximal domains in standard big bang cosmologies in such a way as to preserve certain properties of the metric, thus giving some geometric structure to the big bang singularity and pre-big bang spacetime points.\\  

\noindent The basic idea is  illustrated with the prototype example of the Milne Universe in two spacetime dimensions.  The line element in curvature coordinates is,
\begin{equation}\label{milnemetric}
ds^2=-dt^2+a^2(t)d\chi^2,
\end{equation}
where in this case the scale factor $a(t)=t$. The formulas $\tau=t\cosh\chi$ and $\rho=t\sinh\chi$, with $\tau>|\rho|$, transform $(t,\chi)$ to Fermi coordinates $(\tau,\rho)$ of the comoving observer at $\chi=\rho=0$ \cite{randles}, and the metric in $(\tau,\rho)$ coordinates becomes, 
\begin{eqnarray}
ds^2&=g_{\tau\tau} d\tau^2+d\rho^2\label{milnefermi1}\\
&\equiv-d\tau^2+d\rho^2,\label{milnefermi2}
\end{eqnarray}
i.e., the Milne Universe is diffeomorphic to the interior of the forward lightcone in Minkowski spacetime.  Eq.\eqref{milnefermi1} gives the form of the metric in Fermi coordinates for a general class of scale factors so that in general $\rho$ is proper distance along spacelike geodesics at fixed proper time $\tau$ (see \cite{randles, klein13}).\\  

\noindent In the original curvature coordinates, the metric of Eq.\eqref{milnemetric} is degenerate at the big bang, $t=0$, a coordinate singularity, but this is not the case for Eq.\eqref{milnefermi2}, nor as we will prove for Eq.\eqref{milnefermi1}, provided $\tau>0$, for the case of more general Robertson-Walker cosmologies which have coordinate-independent curvature singularities at $t=0$.\\

\noindent In the $(\tau,\rho)$ Fermi coordinates for the Milne Universe, spacelike geodesics orthogonal to the path of the comoving observer are horizontal straight lines within the forward light cone of Minkowski spacetime as depicted in Fig.\ref{Milne}. The dotted parts of the horizontal line in Fig \ref{Milne} are extensions of the proper distance coordinate, $\rho$, beyond the lightcone boundary of the Milne Universe at cosmological time $t=0$.\\

\begin{figure}[!h]
\begin{center}
\includegraphics[width=0.7\textwidth]{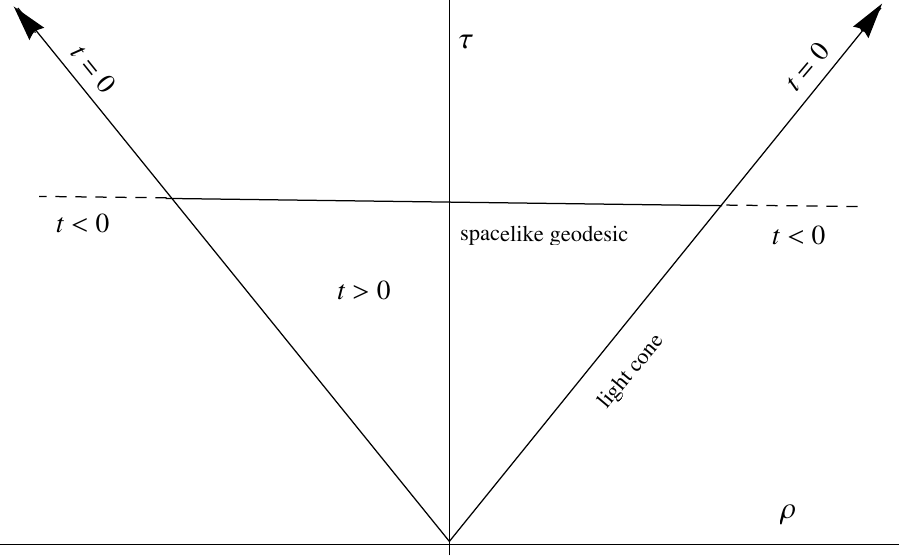}
\end{center}
\caption{The Milne Universe in Fermi coordinates $(\tau,\rho)$ is the interior of the forward light cone in Minkowski space. The comoving observer's worldline is the vertical line $\rho=0$. The dashed portion of the horizontal line extends the spacelike geodesic beyond the Milne Universe to include points with negative cosmological times $t$. }  
\label{Milne}
\end{figure}

\noindent Cosmological time $t$ --- which for notational purposes we shall  designate as $t_{0}$ --- is defined implicitly as a function of $\tau$ and $\rho$ through a natural extension of the inverse Fermi coordinate transformation $t_{0}= t_{0}(\tau,\rho)$ (for the general case see Eq. \eqref{properAlt2} below). In this way, $t_{0}\equiv t<0$ on the dotted portion of the spacelike geodesic in Fig \ref{Milne}.\\

\noindent In this paper, we carry out a similar construction for a class of Robertson-Walker cosmologies consistent with astronomical observations.  In addition to some regularity conditions, we require the scale factor $a(t)$ to be either inflationary near the big bang, or that $\infty>\dot{a}(0)>0$, and in four dimensions that the spacetime is spatially flat.  The cosmological time zero submanifold (defined by $t_{0}=0$) is lightlike in our extension, and parameterized by the (finite) Fermi radius of the orginal universe $\rho_{\mathcal{M}_{\tau}}$ (see Eq.\eqref{radius}).  In general the extension of the metric is not twice continuously differentiable, but continuity is retained along with existence and continuity of the partial derivatives of the nonvanishing leading metric coefficient $g_{\tau\tau}$ (and $g_{\rho\rho}$ is constant). \\

\noindent Our construction is purely geometric and coordinate independent, but Fermi coordinate charts play a useful role because that coordinate system is geometrically constructed.  To define Fermi coordinates, consider a foliation of some neighborhood $\mathcal{U}$ (which might be the entire spacetime) of a comoving observer's  worldline, $\beta(t)$, by disjoint Fermi spaceslices $\{\mathcal{M}_{\tau}\}$.  To define $\mathcal{M}_{\tau}$, let $\varphi_{\tau} : \mathcal{M} \rightarrow \mathbb{R}$ by,
\begin{equation}\label{slice}
\varphi_{\tau}(p)=g(\exp_{\beta(\tau)}^{-1}p,\, \dot{\beta}(\tau)),
\end{equation}
where the overdot represents differentiation with respect to proper time $\tau$ along $\beta$, $g$ is the metric tensor, and the exponential map, $\exp_{p}(v)$ denotes the evaluation at affine parameter $1$ of the geodesic starting at  point $p\in\mathcal{M}$, with initial derivative $v$.  Now define,
\begin{equation}\label{slice2}
\mathcal{M}_{\tau}\equiv \varphi_{\tau}^{-1}(0).
\end{equation}
In other words, the Fermi spaceslice $\mathcal{M}_{\tau}$ of all $\tau$-simultaneous points consists of all the spacelike geodesics orthogonal to the path of the comoving observer $\beta$ at fixed proper time $\tau$.\\

\noindent Fermi coordinates are associated to the foliation $\{\mathcal{M}_{\tau}\}$ in a natural way.  Each spacetime point on $\mathcal{M}_{\tau}$ is assigned time coordinate $\tau$, and the spatial coordinates are defined relative to a parallel transported orthonormal reference frame.   Specifically, a Fermi coordinate system along $\beta$ is determined by an orthonormal frame of vector fields, $e_{0}(\tau), e_{1}(\tau), e_{2}(\tau), e_{3}(\tau)$ parallel along $\beta$, where  $e_{0}(\tau)$ is the four-velocity of the Fermi observer, i.e., the unit tangent vector of $\beta(\tau)$. Fermi coordinates $x^{0}$, $x^{1}$, $x^{2}$, $x^{3}$ relative to this tetrad   are defined by,
\begin{equation}\label{F2}
\begin{split}
x^{0}\left (\exp_{\beta(\tau)} (\lambda^{j}e_{j}(\tau))\right)&= \tau \\
x^{k}\left (\exp_{\beta(\tau)} (\lambda^{j}e_{j}(\tau))\right)&= \lambda^{k}, 
\end{split} 
\end{equation}
where Latin indices  run over $1,2,3$ (and Greek indices run over $0,1,2,3$).\\

\noindent  Fermi coordinates may be constructed in a sufficiently small open neighborhood of any timelike geodesic in any spacetime. The metric tensor expressed in these coordinates is Minkowskian to first order near the geodesic of the Fermi observer, with second order corrections involving only the curvature tensor \cite{MM63}. General formulas in the form of Taylor expansions for coordinate transformations to and from more general Fermi-Walker coordinates are given in \cite{KC1} and exact transformation formulas for a class of spacetimes are given in \cite{CM, KC3}. Applications include the study of relative velocities, tidal dynamics, gravitational waves, statistical mechanics, and the influence of curved space-time on quantum mechanical phenomena 
\cite{Carney1, Carney2, bolos, Jake, CM2, Ishii, pound, FG, KC9, KY, B,P80, PP82, rinaldi}.\\

\noindent It was proved in \cite{randles} that the maximal Fermi chart $(x^{\alpha}, \mathcal{U}_{\mathrm{Fermi}})$ for $\beta(t)$ in a non inflationary\footnote{A Robertson-Walker space-time is non inflationary if $\ddot{a}(t)\leq0$ for all $t$.} Robertson-Walker space-time $(\mathcal{M}, g)$, with increasing scale factor, is global, i.e., $\mathcal{U}_{\mathrm{Fermi}}=\mathcal{M}$.  If, on the other hand, $(\mathcal{M}, g)$ includes inflationary periods, there may exist a cosmological event horizon for the comoving observer, i.e.,
\begin{equation}
\chi_{\mathrm{horiz}}(t_{0})\equiv\int_{t_{0}}^{\infty}\frac{1}{a(t)}dt<\infty,
\end{equation}
for some
for some $t_{0}>0$ (and hence for all $t_{0}>0$).\footnote{$\chi_{\mathrm{horiz}}(t_{0})$ is the $\chi$-coordinate at time $t_{0}$ of the cosmological event horizon, beyond which the co-moving observer at $\chi = 0$ cannot receive a light signal at any future proper time.}  For Robertson-Walker spacetimes with a big bang singularity and a cosmological event horizon, it was proved in \cite{klein13} under a regularity assumption that the maximal Fermi chart $\mathcal{U}_{\mathrm{Fermi}}$ consists of all spacetime points within (but not including) the cosmological event horizon so that the maximal Fermi chart is the causal past of the comoving observer at future infinity.  For cosmologies with no event horizon, it was shown, for both inflationary and non inflationary models, that the Fermi coordinate chart is global.  \\

\noindent It was also shown in \cite{klein13} that all spacelike geodesics with initial point on the worldline $\beta$ of a comoving observer and orthogonal to $\beta$, terminate at the big bang in a finite proper distance $\rho_{\mathcal{M}_{\tau}}$, the radius of $\mathcal{M}_{\tau}$.  In this sense, as already noted by Page \cite{page} using Rindler's observations \cite{rindler}, the big bang is simultaneous with all spacetime events.\\  

\noindent We show in this paper how the spacelike geodesics and the metric tensor can be extended to points in a larger spacetime manifold $\overline{\mathcal{M}}$  with zero or negative cosmological times, analogous to the extensions of the Milne Universe depicted in Fig \ref{Milne}.
For the general case, the extended spacetime\footnote{Here and below, ``extended spacetime'' should be understood to mean ``extended degenerate spacetime,'' in the sense that the spacetime manifold is extended to a larger manifold, but the Lorentzian metric in four spacetime dimensions collapses to a two dimensional Lorentzian metric on the big bang submanifold, to be identified in the sequel.}  $\overline{\mathcal{M}}$ can be expressed as a disjoint union,
\begin{equation}\label{union0}
\overline{\mathcal{M}}=\mathcal{M}^{+}\cup\mathcal{M}^{0}\cup\mathcal{M}^{-},
\end{equation}
where the superscripts indicate respectively that the continuous function $t_{0}(\tau,\rho)$ restricted to the set is positive, zero, or negative. Here $\mathcal{M}^{+}= \mathcal{M}$ denotes the original Robertson-Walker universe. The geometry of the spacetime $\mathcal{M}^{-}$ must be largely undetermined, except for restrictions on the spacetime points close to the big bang, because of our requirement that $g_{\tau\tau}$ be continuously differentiable across $\mathcal{M}^{0}$, that $g_{\rho\rho}\equiv1$ and the remaining metric coefficients in four spacetime dimensions be continuous. The spacetime $\overline{\mathcal{M}}$ can be understood as a smooth manifold with mild singularities of the Lorentzian metric $g$ on $\mathcal{M}^{0}$ representing the big bang. From our extension $\mathcal{M}^{0}$ inherits some geometric structure from the original Robertson-Walker spacetime $\mathcal{M}$.  \\ 

\noindent In two spacetime dimensions, the submanifold $\mathcal{M}^{0}$ defined by $t_{0}(\tau,\rho)=0$ is parameterized in two connected components by $(\tau, \rho_{\mathcal{M}_{\tau}})$, and  $(\tau, -\rho_{\mathcal{M}_{\tau}})$, for $\tau>0$, and we show that $\mathcal{M}^{0}$ is lightlike. In four spacetime dimensions, our extension results in a dimensional reduction of the cotangent bundle at cosmological time zero ($t_{0}=0$) similar to those described in \cite{Carlip, Stoica}.\\

\noindent This paper is organized as follows.  In Section \ref{charts}, we provide a summary of results for maximal Fermi charts on Robertson-Walker cosmologies needed in the sequel.  Section \ref{horizons} reviews relationships and provides a new result on particle and cosmological horizons. This enables us to avoid mutually exclusive conditions on the scale factors we consider.  Section \ref{gtautaulimit} gives results on the limiting values of the metric coefficient $g_{\tau\tau}$ as cosmological time goes to zero, and shows that nonvanishing continuous extensions of $g_{\tau\tau}$ are possible. Section \ref{1+1} is the most technical part of the paper. Here we prove continuity of the partial derivatives of $g_{\tau\tau}$ at the boundary of its domain, the big bang. Section \ref{extend} and Section \ref{3+1} carry out the extensions of the cosmology $\mathcal{M}$ to the larger spacetime $\overline{\mathcal{M}}$ in two and four spacetime dimensions respectively. In Section \ref{examples} we give examples of cosmologies and extensions.  Section \ref{conclusions} summarizes results and offers concluding remarks.  Section \ref{appendix} is the appendix and contains the proofs of the lemmas and theorems in Section \ref{1+1}.

\section{Maximal Fermi Charts}\label{charts}

\noindent This section summarizes results from \cite{randles, klein13}  needed in the sequel.  The Robertson-Walker metric on space-time $\mathcal{M}=\mathcal{M}_{k}$ is given by the line element,

\begin{equation}\label{frwmetric}
ds^2=-dt^2+a^2(t)\left[d\chi^2+S^2_k(\chi)d\Omega^2\right],
\end{equation}

\noindent where $d\Omega^2=d\theta^{2}+\sin^{2}\theta \,d\varphi^{2}$, $a(t)$ is the scale factor, and, 

\begin{equation}\label{Sk}
S_{k}(\chi)=
\begin{cases}
\sin\chi &\text{if}\,\, k= 1\\
 \chi&\text{if}\,\,k= 0\\ 
 \sinh\chi&\text{if}\,\,k=-1.
\end{cases}
\end{equation}

\noindent The coordinate $t>0$ is cosmological time and $\chi, \theta, \varphi$ are dimensionless. Here $\theta$ and $\phi$ lie in intervals $I_{\pi}$ and $I_{2\pi}$ of lengths $\pi$ and $2\pi$ respectively. The values $+1,0,-1$ of the parameter $k$ distinguish the three possible maximally symmetric space slices for constant values of $t$ with positive, zero, and negative curvatures respectively.  The radial coordinate $\chi$ takes all positive values for $k=0$ or $-1$, but is bounded above by $\pi$ for $k=+1$.\\

\noindent We assume henceforth  that $k= 0$ or $-1$ so that the range of $\chi$ is unrestricted. The techniques needed for the case $k=+1$ are the same, but require the additional restriction that $\chi<\pi$ so that spacelike geodesics do not intersect.  We note that $k=+1$ for the Einstein static universe, for which  Fermi coordinates for geodesic observers are global (except for the antipode, $\chi = \pi$) \cite{KC3}.\\

\noindent We assume throughout this paper that the scale factor $a(t)$ is \emph{regular}, i.e., it satisfies the following definition \cite{klein13}.

\begin{definition}\label{regular} Define the scale factor $a(t):[0,\infty)\rightarrow[0,\infty)$ to be \emph{regular} if: 
\begin{enumerate}
\item [(a)] $a(0)=0$, i.e., the associated cosmological model includes a big bang.
\item[(b)] $a(t)$ is increasing and continuous on $[0,\infty)$ and twice continuously differentiable on $(0,\infty)$, with inverse function $b(t)$ on $[0,\infty)$.
\item[(c)] For all $t>0$,
\begin{equation}\label{condition}
\frac{a(t)\ddot{a}(t)}{\dot{a}(t)^{2}}\leq 1.
\end{equation}
\end{enumerate}
If in addition the expression in Eq. \eqref{condition} is bounded below by a constant $-K\leq -1$, we call the scale factor $a(t)$ \emph{strongly regular}.
\end{definition}

\begin{example}\label{2/3}
It is easily verified that power law scale factors of the form $a(t)=t^\alpha$ are strongly regular for all $\alpha>0$. Scale factors of this form include radiation and matter dominated universes as well as inflationary universes for the cases $\alpha>1$. Similarly the inflationary scale factor $a(t) = \sinh t$ is easily seen to be strongly regular.  A more elaborate example of a strongly regular scale factor is given by Eq.\eqref{lambdamatter} and is discussed below in Section \ref{examples}.
\end{example}

\begin{remark} \label{regineq} Under the assumption that $\dot{a}(t)>0$ for all $t$, and $a(t)$ is regular, it follows that for any $\tau>0$, there exists $t\in(0,\tau)$ such that the inequality Eq\eqref{condition} is strict at $t$, and hence by continuity, on an open interval containing $t$.  This follows from the observation that equality in Eq.\eqref{condition} forces $a(t)$ to be an exponential function which violates Definition \ref{regular}a.
\end{remark}

\noindent Let $\beta$ be the path of the comoving observer with fixed coordinate $\chi=0$.  As a preliminary step to express the metric of Eq.\eqref{frwmetric} in Fermi coordinates of $\beta$, we define coordinate transformation functions. For $\tau>t_{0}>0$, define \cite{klein13, Bolos12}  
 \begin{equation}\label{key2}
\chi_{t_{0}}(\tau)=\int_{t_{0}}^{\tau}\frac{1}{a(t)}\frac{a(\tau)}{\sqrt{a^{2}(\tau)-a^{2}(t)}}\,dt.
\end{equation}
The function $\chi_{t_{0}}(\tau)$ is the value of the $\chi$-coordinate of the spacetime point with $t$-coordinate $t_{0}$ on the spacelike geodesic orthogonal to $\beta$ with initial point $\beta(\tau)$. The proper distance $\rho$ along that geodesic from $\beta(\tau)$ to the point with $t$-coordinate $t_{0}$ is given by, 
\begin{equation}\label{properAlt}
\rho=\int_{t_{0}}^{\tau}\frac{a(t)}{\sqrt{a^{2}(\tau)-a^{2}(t)}}\,dt.
\end{equation}
The proper distance along the geodesic increases as cosmological time $t_{0}$ decreases monotonically to zero.  The radius, $\rho_{\mathcal{M}_{\tau}}$, of the Fermi spaceslice $\mathcal{M}_{\tau}$ (see Eq.\eqref{slice2}) is the proper distance along the spacelike geodesic orthogonal to the comoving observer $\beta(\tau)$, from $\beta(\tau)$ to the big bang at $t =0$.  It is given by,

\begin{equation}\label{radius}
\rho_{\mathcal{M}_{\tau}}=\int_{0}^{\tau}\frac{a(t)}{\sqrt{a^{2}(\tau)-a^{2}(t)}}\,dt.
\end{equation}
\noindent It is easy to show, \cite{klein13}, that for a regular scale factor,

\begin{equation}\label{radiusbound}
\rho_{\mathcal{M}_{\tau}} \leq  \frac{\pi}{2}\frac{1}{H(\tau)},
\end{equation}
where $H(\tau) = \dot{a}(\tau)/a(\tau)$ is the Hubble parameter. Moreover, if $a(t)$ is strongly regular, then
\begin{equation}\label{dradius}
\der{\rho_{\mathcal{M}_\tau}}{\tau}(\tau)=\frac{\dot{a}(\tau)}{a(\tau)}\int_0^\tau\left(1-\frac{a(t)\ddot{a}(t)}{\dot{a}^2(t)}\right)\frac{a(t)dt}{\sqrt{a^2(\tau)-a^2(t)}}>0,
\end{equation}
so that $\rho_{\mathcal{M}_\tau}$ is an increasing function of $\tau$ (i.e. the Fermi radius of the universe increases with $\tau$).\footnote{A correction to the published theorem giving this result is posted on arXive, see \cite{klein13}.}  Denote the Fermi coordinates for the comoving observer $\beta(\tau)=(\tau,0,0,0)$ by $\{\tau, x=x^1, y=x^2, z=x^3\}$ according to Eq. \eqref{F2}.
Under the assumption that $a(t)$ is regular, the maximal Fermi chart $\mathcal{U}_{\mathrm{Fermi}}\subset \mathcal{M}$ is given by,
\begin{equation}\label{UFermi} 
\mathcal{U}_{\mathrm{Fermi}} = \left\{(\tau,x,y,z): \tau >0 \text{ and }  \sqrt{x^2+y^2+z^2} < \rho_{\mathcal{M}_{\tau}}\right\},
\end{equation}
and the metric in Fermi coordinates is given by,
\begin{equation}
\begin{split}\label{fermimetric}
ds^2=&\,g_{\tau\tau}(\tau,\rho)\, d\tau^2+dx^2 +dy^2+dz^2\\ 
+&\lambda_{k}(\tau,\rho)\big[(y^2+z^2)dx^2+(x^2+z^2)dy^2+(x^2+y^2)dz^2\\
-&xy(dxdy+dydx)-xz(dxdz+dzdx)-yz(dydz+dzdy)\big],
\end{split}
\end{equation} 
where $\rho=\sqrt{x^2+y^2+z^2}$, 
\begin{equation}
\begin{split}\label{oldg}
\quad g_{\tau\tau}(\tau,\rho)&=-\dot{a}(\tau)^{2}\left[a^{2}(\tau)-a^{2}(t_{0})\right]\times\\
&\left[\frac{1}{\dot{a}(t_{0})\sqrt{a^{2}(\tau)-a^{2}(t_{0})}}-\int_{t_{0}}^{\tau}\frac{\ddot{a}(t)}{\dot{a}(t)^{2}}\frac{dt}{\sqrt{a^{2}(\tau)-a^{2}(t)}}\right]^{2},\\
\end{split}
\end{equation}
and,

\begin{equation}\label{lambda2}
\lambda_{k}(\tau,\rho) =\frac{a^{2}(t_{0})S^2_k(\chi_{t_{0}}(\tau))-\rho^{2}}{\rho^{4}},
\end{equation}
for $\rho>0$ and $\lambda_{k}(\tau,0)=0$. Here, $t_{0}=t_{0}(\tau,\rho)$ is defined implicitly by Eq.\eqref{properAlt},\footnote{The subscript on the cosmological time coordinate $t_{0}$ is included as a convenience so that we may use the symbol $t$ as a dummy variable in integral expressions where it arises naturally.} and $S_k$ is given by Eq.\eqref{Sk}.
It may be shown \cite{randles}, that $\lambda_{k}(\tau,\rho)$ is a smooth function of $\tau$ and $\rho^2$.\\ 

\noindent Applying a standard transformation from Cartesian to spherical coordinates in $\mathbb{R}^{3}$ to the Fermi space coordinates results in the diagonal metric for \emph{Fermi polar coordinates},
\begin{equation}\label{fermipolar}
ds^2=g_{\tau\tau} d\tau^2+d\rho^2 + a^{2}(t_{0})S^2_k(\chi_{t_{0}}(\tau))d\Omega^2,
\end{equation}
with Fermi chart,
\begin{equation}\label{fermipolarchart}
\mathcal{U}_{\mathrm{polar}}=\{(\tau,\rho,\theta,\phi): \tau>0, 0<\rho<\rho_{\mathcal{M}_{\tau}}, \theta\in I_{\pi},\, \phi\in I_{2\pi}\}
\end{equation}

\begin{remark}\label{Mink2} In the Milne Universe where $k=-1$ and $a(t)=t$, it is easily verified that $g_{\tau\tau}\equiv1$ and,
\begin{equation}
a(t_{0})S_k(\chi_{t_{0}}(\tau))= \rho,
\end{equation}
where $\rho=\rho(\tau,t_{0})$ according to Eq.\eqref{properAlt}. Therefore,
\begin{equation}
\lim_{t_{0}\to0^{+}}a(t_{0})S_k(\chi_{t_{0}}(\tau))=\lim_{\quad\rho\to\rho_{\mathcal{M}_\tau}^{-}}\rho=\rho_{\mathcal{M}_\tau}=\tau.
\end{equation}
Then from Eq.\eqref{lambda2},
\begin{equation}
\lambda_{k}(\tau,\rho)\equiv0.
\end{equation}
so  Fermi coordinates in the Milne Universe are just the usual Minkowski coordinates.
\end{remark}

\section{Cosmological and particle horizons}\label{horizons}

\noindent In this section we collect and  prove results that relate the existence of particle horizons and cosmological event horizons to properties of the scale factor $a(t)$ and its derivatives. A Robertson-Walker spacetime has a cosmological event horizon if,
\begin{equation}
\chi_{\mathrm{horiz}}(t_{0})\equiv\int_{t_{0}}^{\infty}\frac{1}{a(t)}dt<\infty,
\end{equation}
for some $t_{0}>0$ (and hence all $t_{0}>0$). 
The spacetime has has finite particle horizon if,
\begin{equation}
\chi_{\text{part}}(\tau)\equiv\int_0^\tau\frac{1}{a(t)}dt<\infty,
\end{equation}  
for some $\tau>0$ (and hence all $\tau>0$). Part (b) of the following theorem shows that a finite particle horizon is mathematically impossible if $\dot{a}(0^{+}) <\infty$.

\begin{theorem}\label{inflation}
Let $a(t)$ be a regular scale factor on a Robertson-Walker spacetime $(\mathcal{M},g)$, where $g$ is given by Eq.\eqref{frwmetric}.
\begin{enumerate}
\item[(a)] If $\mathcal{M}$ has a cosmological horizon, i.e., $\chi_{\mathrm{horiz}}(t_{0})<\infty$ for some $t_{0}>0$, then 
\begin{equation}
\lim_{t\rightarrow\infty} \frac{t}{a(t)}=0=\lim_{t\rightarrow\infty} \dfrac{1}{\dot{a}(t)}.
\end{equation}
Moreover, $\mathcal{M}$ experiences inflationary periods for arbitrarily large cosmological times, that is, for any $N>0$, there exists a non empty open interval $(a,b)$ with $a>N$ such that $\ddot{a}(t)>0$ on $(a,b)$.  However, the condition $\ddot{a}(t)>0$ for all $t>0$ does not imply the existence of a cosmological event horizon.
\item[(b)] If $\mathcal{M}$ has a finite particle horizon, i.e., $\chi_{\text{part}}(\tau)<0$ for some $\tau>0$, then
\begin{equation}\label{finpartlim}
\lim_{t\to0^+}\frac{t}{a(t)}=0=\lim_{t\to0^+}\frac{1}{\dot{a}(t)}.
\end{equation}
Moreover, $\mathcal{M}$ experiences noninflationary periods for arbitrarily small cosmological times, that is, for any $\delta>0$, there exists a non empty open interval $(a,b)\subseteq (0,\delta)$ such that $\ddot{a}(t)<0$ on $(a,b)$. 
\end{enumerate}
\end{theorem}

\begin{proof}
The proof of part (a) is given in \cite{klein13}. To prove part (b), observe first that the right hand side of Eq. \eqref{finpartlim} follows from the left hand side by L'H\^{o}pital's rule. Observe that
\begin{equation}\label{intlim1}
\lim_{t_0\to0^+}\int_{t_0}^\tau \frac{1}{t}\frac{1}{a(t)}dt=\infty,
\end{equation}
because $\lim_{t\to0^+}1/a(t)=\infty$. Also, for $\tau>t_0>0$,
\begin{equation}
\left|I_{[t_0,\tau]}\frac{t_0}{t}\frac{1}{a(t)}\right|\leqslant \frac{1}{a(t)},
\end{equation}
for $t>0$, where $I_{[t_0,\tau]}$ is the indicator function for the interval $[t_0,\tau]$. So, from the Lebesgue Dominated convergence theorem,
\begin{equation}\label{intlim2}
\lim_{t_0\to0^+}t_0\int_{t_0}^\tau\frac{1}{t}\frac{1}{a(t)}dt=\lim_{t_0\to0^+}\int_0^\tau I_{[t_0,\tau]}\frac{t_0}{t}\frac{1}{a(t)}dt=0.
\end{equation}
Now using Eqs. \eqref{intlim1}, \eqref{intlim2} and L'H\^{o}pital's rule we have that
\begin{multline}
\lim_{t_0\to0^+}t_0\int_{t_0}^\tau\frac{1}{t}\frac{1}{a(t)}dt=\lim_{t_0\to0^+}\left(\int_{t_0}^\tau\frac{1}{t}\frac{1}{a(t)}dt\right)/\left(\frac{1}{t_0}\right)\\=\lim_{t_0\to0^+}-\left(\frac{1}{t_0}\frac{1}{a(t_0)}\right)/\left(-\frac{1}{t_0^2}\right)=\lim_{t_0\to0^+}\frac{t_0}{a(t_0)}=0.
\end{multline}
Given $\delta>0$, $\lim_{t\to0^+}\dot{a}(t)=\infty$ implies that $\dot{a}(t)$ cannot be increasing on $(0,\delta)$. Therefore there must be a $t\in (0,\delta)$ with $\ddot{a}(t)<0$, and by continuity $\ddot{a}(t)<0$ on an open interval containing $t$. 
\end{proof}

\begin{remark} Part b of Theorem \ref{inflation} shows that a scale factor $a(t)$ with $\dot{a}(0^+) =\infty$ such as $a(t)=t^{\alpha}$ for $\alpha<1$ is non inflationary at the big bang.  Our results for continuously differentiable extensions of the metric coefficient $g_{\tau\tau}$ therefore exclude such scale factors, but for continuous extensions of the Fermi metric for this case, see Theorems \ref{gcontin} and \ref{justcontin} below.
\end{remark}

\section{Limiting values of $g_{\tau\tau}$ at the big bang}\label{gtautaulimit}

In this section, we find the limiting value of $g_{\tau\tau}(\tau,\rho)$ as $\rho$ increases to $\rho_{\mathcal{M}_{\tau}}$, and show that the limiting value is not zero.  This will be used in Section \ref{1+1} to construct a continuously differentiable extension of $g_{\tau\tau}$ to a larger spacetime $\overline{\mathcal{M}}$ (see Eq. \eqref{union}).  We begin with an alternative but equivalent expression for $g_{\tau\tau}(\tau,\rho)$ in Eq. \eqref{oldg}.  Rearranging terms gives,
\begin{equation}
\label{newg}
g_{\tau\tau}(\tau,\rho)= -\dot{a}(\tau)^{2}\left[\frac{1}{\dot{a}(t_{0})}-\int_{t_{0}}^{\tau}\frac{\ddot{a}(t)}{\dot{a}(t)^{2}}\frac{\sqrt{a^{2}(\tau)-a^{2}(t_{0})}}{\sqrt{a^{2}(\tau)-a^{2}(t)}}dt\right]^{2}
\end{equation}
 Eq.\eqref{oldg} may then be modified by substituting the relation,

\begin{equation}
\frac{1}{\dot{a}(t_{0})}=\frac{1}{\dot{a}(\tau)}+\int_{t_{0}}^{\tau}\frac{\ddot{a}(t)}{\dot{a}(t)^{2}}dt.
\end{equation}
The result is,
\begin{equation}
\label{newg2}
g_{\tau\tau}(\tau,\rho)= -\left[1-\dot{a}(\tau)\int_{t_{0}}^{\tau}\frac{\ddot{a}(t)}{\dot{a}(t)^{2}}\left(\frac{\sqrt{a^{2}(\tau)-a^{2}(t_{0})}}{\sqrt{a^{2}(\tau)-a^{2}(t)}}-1\right)dt\right]^{2}.
\end{equation}
The following lemma slightly generalizes a result in \cite{klein13} to the case $t_{0}=0$.
\begin{lemma}\label{intbound} Let $a(t)$ be a regular scale factor.  Then for $0\leq t_{0}<\tau$,

\begin{equation}\label{dchi}
\begin{split}
\int_{t_{0}}^{\tau}\frac{\ddot{a}(t)}{\dot{a}(t)^{2}}\left[\frac{a(\tau)}{\sqrt{a^{2}(\tau)-a^{2}(t)}}-1\right]dt&<\int_{t_{0}}^{\tau}\frac{1}{a(t)}\left[\frac{a(\tau)}{\sqrt{a^{2}(\tau)-a^{2}(t)}}-1\right]dt\\
&< \frac{1}{\dot{a}(\tau)}.
\end{split}
\end{equation}\\
\end{lemma}

\begin{proof} The first inequality follows from Definition \ref{regular}c and Remark \ref{regineq}.  For the second inequality,

\begin{equation}\label{long}
\begin{split}
\int_{t_{0}}^{\tau}\frac{1}{a(t)}&\left[\frac{a(\tau)}{\sqrt{a^{2}(\tau)-a^{2}(t)}}-1\right]dt=\int_{t_{0}}^{\tau}\frac{1}{a(t)}\left[\frac{1-\sqrt{1-\frac{a^{2}(t)}{a^{2}(\tau)}}}{\sqrt{1-\frac{a^{2}(t)}{a^{2}(\tau)}}}\right]dt
\\
=&\frac{1}{a^{2}(\tau)}\int_{t_{0}}^{\tau}\frac{a(t)}{\sqrt{1-\frac{a^{2}(t)}{a^{2}(\tau)}}\left(1+\sqrt{1-\frac{a^{2}(t)}{a^{2}(\tau)}}\right)}dt\\
=&\frac{1}{a(\tau)}\int_{t_{0}}^{\tau}\frac{a(t)}{\dot{a}(t)}\frac{\dot{a}(t)/a(\tau)}{\sqrt{1-\frac{a^{2}(t)}{a^{2}(\tau)}}\left(1+\sqrt{1-\frac{a^{2}(t)}{a^{2}(\tau)}}\right)}dt\\
<&\frac{1}{a(\tau)}\frac{a(\tau)}{\dot{a}(\tau)}\int_{0}^{\tau}\frac{\dot{a}(t)/a(\tau)}{\sqrt{1-\frac{a^{2}(t)}{a^{2}(\tau)}}\left(1+\sqrt{1-\frac{a^{2}(t)}{a^{2}(\tau)}}\right)}dt,\\
\end{split}
\end{equation}
where in the last step, we have used the fact that the Hubble parameter, $H(t)$, is a decreasing function of $t$, and strictly decreasing on an interval, so that $a(t)/\dot{a}(t)<a(\tau)/\dot{a}(\tau)$ on some interval of $t$ values.  To evaluate this last integral, we make the change of variable, $x=a(t)/a(\tau)$, which yields,
\begin{equation}\label{integral}
\begin{split}
\int_{t_{0}}^{\tau}\frac{1}{a(t)}\left[\frac{a(\tau)}{\sqrt{a^{2}(\tau)-a^{2}(t)}}-1\right]dt&<\frac{1}{\dot{a}(\tau)}\int_{0}^{1}\frac{dx}{\sqrt{1-x^{2}}(1+\sqrt{1-x^{2}})}=\frac{1}{\dot{a}(\tau)}.
\end{split}
\end{equation}
\end{proof}

\begin{theorem}\label{c0thrm} 
Let $a(t)$ be strongly regular. Then the Fermi metric coefficient $g_{\tau\tau}(\tau,\rho)$ satisfies the following:
\begin{equation}\label{glimit}
\lim_{\quad\rho\rightarrow\rho_{\mathcal{M}_{\tau}}^{-}}g_{\tau\tau}(\tau,\rho)= -\left[1-\dot{a}(\tau)\int_{0}^{\tau}\frac{\ddot{a}(t)}{\dot{a}(t)^{2}}\left(\frac{a(\tau)}{\sqrt{a^{2}(\tau)-a^{2}(t)}}-1\right)dt\right]^{2}
\end{equation}
Moreover, the function $g_{\tau\tau}(\tau,\rho_{\mathcal{M}_{\tau}})$ defined by this limit satisfies,
\begin{equation}
\infty > -g_{\tau\tau}(\tau,\rho_{\mathcal{M}_{\tau}})>0.
\end{equation}
\end{theorem}

\begin{proof} Note that for fixed $\tau$, $\rho\rightarrow\rho_{\mathcal{M}_{\tau}}$ if and only if $t_{0}\rightarrow 0$.  The existence of the finite limit in Eq \eqref{glimit} follows from Lemma \ref{intbound} and the Dominated Convergence Theorem using the comparison,

\begin{equation}
\frac{|\ddot{a}(t)|}{\dot{a}(t)^{2}}\left[\frac{\sqrt{a^{2}(\tau)-a^{2}(t_{0})}}{\sqrt{a^{2}(\tau)-a^{2}(t)}}-1\right]<\frac{K}{a(t)}\left[\frac{a(\tau)}{\sqrt{a^{2}(\tau)-a^{2}(t)}}-1\right]
\end{equation}

\noindent The second assertion, $-g_{\tau\tau}(\tau,\rho_{\mathcal{M}_{\tau}})>0$, follows directly from Lemma \ref{intbound}.

\end{proof}

\begin{remark} It follows from Eq. \eqref{newg2} that $-g_{\tau\tau}(\tau,\rho)\leq1$ if the cosmological time interval from $t_{0}(\tau,\rho)$ to $\tau$ is inflationary, i.e. if $\ddot{a}(t)\geq 0$ on that interval.  Similarly, $-g_{\tau\tau}(\tau,\rho)\geq1$ if the cosmological time interval from $t_{0}(\tau,\rho)$ to $\tau$ is noninflationary, i.e. if $\ddot{a}(t)\leq 0$ on that interval.
\end{remark}

\noindent The following purely technical lemma will be needed in the proof of continuity of the extension of $g_{\tau\tau}$ and its derivatives.

\begin{lemma}\label{continuitylemma} 
Let $a(t)$ be a regular scale factor and assume that $\tau\geq\tau_{0}>0$. Let $\ell(t)$ be a smooth function defined for $t>0$ satisfying
\begin{equation}
|\ell(t)|<\frac{K}{a(t)}
\end{equation}
for some $K>0$ and all $t>0$. Then,
\begin{enumerate}
\item[(a)] 
If $0<x\leq\tau_{0}$,
\begin{equation}
\int_{\tau_{0}}^{\tau}|\ell(t)|\left[\frac{\sqrt{a^{2}(\tau)-a^{2}(x)}}{\sqrt{a^{2}(\tau)-a^{2}(t)}}-1\right]dt<\frac{K}{\dot{a}(\tau)}\left[1-\left(\frac{a(\tau)}{a(\tau_{0})}- \sqrt{\frac{a^{2}(\tau)}{a^{2}(\tau_{0})}-1}  \right) \right]\\
\end{equation}\\
\item[(b)]
\begin{multline}
a(\tau)\int_{\tau_{0}}^\tau |\ell(t)|\frac{dt}{\sqrt{a^2(\tau)-a^2(t)}}<K\int_{\tau_{0}}^{\tau}\frac{1}{a(t)}dt\\
+\frac{K}{\dot{a}(\tau)}\left[1-\left(\frac{a(\tau)}{a(\tau_{0})}- \sqrt{\frac{a^{2}(\tau)}{a^{2}(\tau_{0})}-1}  \right) \right]
\end{multline}
\end{enumerate}
\end{lemma}

\begin{proof} Using Eq.\eqref{long} we have,

\begin{multline}\label{long2}
\int_{\tau_{0}}^{\tau}\frac{1}{a(t)}\left[\frac{a(\tau)}{\sqrt{a^{2}(\tau)-a^{2}(t)}}-1\right]dt\\
<\frac{1}{\dot{a}(\tau)}\int_{\tau_{0}}^{\tau}\frac{\dot{a}(t)/a(\tau)}{\sqrt{1-\frac{a^{2}(t)}{a^{2}(\tau)}}\left(1+\sqrt{1-\frac{a^{2}(t)}{a^{2}(\tau)}}\right)}dt,
\end{multline}
The change of variable, $x=a(t)/a(\tau)$, yields

\begin{equation}\label{integral}
\begin{split}
\int_{\tau_{0}}^{\tau}\frac{1}{a(t)}\left[\frac{a(\tau)}{\sqrt{a^{2}(\tau)-a^{2}(t)}}-1\right]dt&<\frac{1}{\dot{a}(\tau)}\int_{\frac{a(\tau_{0})}{a(\tau)}}^{1}\frac{dx}{\sqrt{1-x^{2}}(1+\sqrt{1-x^{2}})}\\
&=\frac{1}{\dot{a}(\tau)}\left[1-\left(\frac{a(\tau)}{a(\tau_{0})}- \sqrt{\frac{a^{2}(\tau)}{a^{2}(\tau_{0})}-1}  \right) \right]\\
\end{split}
\end{equation}
Part (b) now follows by rearranging terms and from the hypothesis 

\begin{equation}\label{bdbelow}
|\ell(t)|<\frac{K}{a(t)}
\end{equation}
Part (a) follows from this hypothesis and since $a(\tau)\geq\sqrt{a^{2}(\tau)-a^{2}(x)}$.
\end{proof}

\section{$\mathcal{C}^{1}$ extension of $g_{\tau\tau}$}\label{1+1}

\noindent The plan of this section is first to extend $g_{\tau\tau}$, given by Eq. \eqref{newg2}, as a continuous function to values of $\rho > \rho_{\mathcal{M}_{\tau}}$ (or equivalently to negative values of $t_{0}$).  We then show under some regularity conditions that the first partial derivatives of the extension of $g_{\tau\tau}$ are also continuous.\\

\noindent In order to accomplish this, we must extend the domain of the scale factor to include negative values of cosmological time $t=t_{0}$.  Although not essential, it is convenient using our methods to extend $a(t)$ as an even function so that $a(-t)=a(t)$.  A smooth even extension of $a(t)$ requires $\dot{a}(0)=0$, but we also consider the possibility that $0<\dot{a}(0^{+})<\infty$, which forces a discontinuity in the first derivative of $a(t)$ (but not in the metric coefficients). Both of these geometric properties of the big bang are of interest and result in different extensions of Robertson-Walker spacetimes to negative cosmological times.  As a convenience to the reader, the proofs of the lemmas and theorems of this section have been placed in the appendix labeled as Sect.\ref{appendix}. \\

\noindent As a first step, we extend the proper distance coordinate $\rho$ by the same formula as Eq. \eqref{properAlt}, but so as to allow values of $t_{0}$ in the interval $-\tau<t_0<\tau$,
\begin{equation}\label{properAlt2}
\rho=\int_{t_{0}}^{\tau}\frac{a(t)}{\sqrt{a^{2}(\tau)-a^{2}(t)}}\,dt.
\end{equation} 
Now $t_0 (\tau,\rho)$ is defined implicitly by Eq.\eqref{properAlt2} on the set,
\begin{equation}\label{D}
D=\{(\tau,\rho): \tau>0, 0<\rho<2\rho_{\mathcal{M}_{\tau}}\},
\end{equation}
where $\rho_{\mathcal{M}_{\tau}}$ is the Fermi radius of the universe at proper time $\tau$ of the central observer and is given by Eq.\eqref{radius}.\\

\noindent It follows from the Implicit Function theorem that the function $t_{0}(\tau, \rho)$ is a smooth function of its arguments in $D$, except possibly when $t_{0}=0$, the cosmological time coordinate of the big bang. The next lemma shows that $t_{0}$ is continuous even where $t_{0}=0$. This result will be needed in what follows.  
\begin{lemma}\label{t0cont}
Suppose $a(t)$ is strongly regular. Then the function $t_{0}(\tau,\rho)$, defined implicitly by Eq. \eqref{properAlt2}, is continuous on $D$.
\end{lemma}

\noindent In Fermi polar coordinates, there is a coordinate singularity at $\rho=0$, but this singularity disappears in $1+1$ dimensions and $\rho$ may be extended symmetrically to negative values as well. However, in what follows it is convenient to restrict $\rho$ to nonnegative values, and this causes no loss of generality. \\


%
%

\noindent With Eq. \eqref{newg2} in mind and with a slight abuse of notation, we define the extension of the metric tensor to $D$ by $g_{\rho\rho}\equiv1$ and,
\begin{equation}
\label{newg3}
g_{\tau\tau}(\tau,\rho)= -[1-\dot{a}(\tau)f(\tau,t_0(\tau,\rho))]^{2}\equiv -[1-\dot{a}(\tau)\mathbf{f}(\tau,\rho)]^{2},
\end{equation}

\noindent where

\begin{equation}\label{gint}
f(\tau,t_0)=\begin{cases}
\int_{t_{0}}^{\tau}\frac{\ddot{a}(t)}{\dot{a}(t)^{2}}\left(\frac{\sqrt{a^{2}(\tau)-a^{2}(t_{0})}}{\sqrt{a^{2}(\tau)-a^{2}(t)}}-1\right)dt&\text{if }\quad\tau>t_0\geqslant 0\\
2f(\tau,0)-f(\tau,-t_0)&\text{if }-\tau<t_0<0,
\end{cases}
\end{equation}
and
\begin{equation}\label{f(tau,rho)}
\mathbf{f}(\tau,\rho)\equiv f(\tau,t_0(\tau,\rho)).
\end{equation}

\begin{remark}
Since $\lim_{t_0\to0^+}f(\tau,t_0)=f(\tau,0)$ exists by Theorem \ref{c0thrm}, by the definition of $f$ we automatically have that $\lim_{t_0\to0^-}f(\tau,t_0)=f(\tau,0)$. Therefore $\lim_{t_0\to0}f(\tau,t_0)=f(\tau,0)$.  We note also that for $t_{0}<0$ an equivalent expression for $f(\tau,t_{0})$ is,
\begin{equation}
f(\tau,t_0)= f(\tau,0)+\int^{0}_{t_{0}}\frac{\ddot{a}(t)}{\dot{a}(t)^{2}}\left(\frac{\sqrt{a^{2}(\tau)-a^{2}(t_{0})}}{\sqrt{a^{2}(\tau)-a^{2}(t)}}-1\right)dt.
\end{equation}
\end{remark}

\begin{remark} \label{toobig}
We note that in general the domain $D$  given by Eq.\eqref{D} will be too large to be a coordinate chart for an extended spacetime $\overline{\mathcal{M}}$, and a proper subset must be used to avoid zeros of $g_{\tau\tau}$ in $\mathcal{M}^{-}$ where $t_{0}<0$. This because the conclusions of Theorem \ref{c0thrm} do not necessarily hold for Eq. \eqref{newg3} on all of $D$.  However, it follows from the continuity of $g_{\tau\tau}$ established in Theorem \ref{gcontin} that $-g_{\tau\tau}(\tau,\rho)>0$ at all points $(\tau,\rho)$ where $t_{0}(\tau,\rho)$ is sufficiently close to zero, including points where $t_{0}(\tau,\rho)<0$.   We elaborate further in Section \ref{extend}.
\end{remark}

%
%

\begin{theorem}\label{gcontin} Let $a(t)$ be strongly regular. Then the metric coefficient $g_{\tau\tau}(\tau,\rho)$ given by Eq.\eqref{newg3} is continuous on $D$.
\end{theorem}

\noindent Our next task is to prove that the metric coefficient, $g_{\tau\tau}$, as defined by Eq. \eqref{newg3} is differentiable on $D$.  The following lemma deals with the technicality of the unbounded integrand in Eq. \eqref{gint}.

\begin{lemma}\label{c1lemma1} Let $a(t)$ be regular with $\dot{a}(0^+)<\infty$. The function $f(\tau,t_0)$ given by Eq. \eqref{gint} is continuously differentiable with respect to $t_0\neq0$ when $-\tau<t_0<\tau$, and,
\begin{equation}\label{dt0g}
\partial_{t_0}f(\tau,t_0)=-\frac{a(t_{0})\dot{a}(|t_{0}|)}{\sqrt{a^{2}(\tau)-a^{2}(t_{0})}}\int_{|t_{0}|}^{\tau}\frac{\ddot{a}(t)}{\dot{a}(t)^{2}}\frac{dt}{\sqrt{a^{2}(\tau)-a^{2}(t)}}.
\end{equation}
\end{lemma}

\begin{lemma}\label{c1lemma2}
For a regular scale factor $a(t)$, let $a(0^+)\equiv\lim_{t\to0^+} \dot{a}(t)$. Then:
\begin{enumerate}
\item[(a)] If $\dot{a}(0^+)=0$ and there exists an $\epsilon>0$ such that $\ddot{a}(t)\geqslant0$ for $t\in(0,\epsilon)$, then
\begin{equation}\label{divrint}
\lim_{t_0\to0^+}\int_{t_0}^\tau \frac{\ddot{a}(t)}{\dot{a}^2(t)}\frac{dt}{\sqrt{a^2(\tau)-a^2(t)}}=\infty.
\end{equation}
\item[(b)] If $\infty>\dot{a}(0^+)>0$ then,
\begin{equation}\label{convint}
\lim_{t_0\to0^+}\int_{t_0}^\tau \frac{\ddot{a}(t)}{\dot{a}^2(t)}\frac{dt}{\sqrt{a^2(\tau)-a^2(t)}}=\int_{0}^\tau \frac{\ddot{a}(t)}{\dot{a}^2(t)}\frac{dt}{\sqrt{a^2(\tau)-a^2(t)}}
\end{equation}
exists and is finite.
\end{enumerate}
\end{lemma}

\begin{theorem}\label{drhometric'}
Let $a(t)$ be a strongly regular scale factor and suppose that one of the conditions of Lemma \ref{c1lemma2} holds. Then $g_{\tau\tau}$ is  differentiable with respect to $\rho$ in $D$, and
\begin{equation}\label{drhometric}
\partial_\rho g_{\tau\tau}=2\sqrt{-g_{\tau\tau}}\dot{a}(\tau)\partial_\rho \mathbf{f}(\tau,\rho),
\end{equation}
where
\begin{equation}\label{drhometric2}
\partial_\rho \mathbf{f}(\tau,\rho)=\dot{a}(|t_0|)\int_{|t_{0}|}^{\tau}\frac{\ddot{a}(t)}{\dot{a}(t)^{2}}\frac{dt}{\sqrt{a^{2}(\tau)-a^{2}(t)}}
\end{equation}
for $\rho\neq \rho_{\mathcal{M}_\tau}$ (and with $t_{0}=t_{0}(\tau,\rho)$), and
\begin{equation}\label{casedrhog}
\partial_\rho \mathbf{f}(\tau,\rho_{\mathcal{M}_\tau})=
\begin{cases}
\dot{a}(0^+)\int_0^\tau\frac{\ddot{a}(t)}{\dot{a}^2(t)}\frac{1}{\sqrt{a^2(\tau)-a^2(t)}}dt& \text{ if } \dot{a}(0^+)>0\\
\frac{1}{a(\tau)}&\text{ if } \dot{a}(0)=0
\end{cases}
\end{equation}
\end{theorem}

\noindent We next establish the differentiability of $g_{\tau\tau}(\tau,\rho)$ with respect to $\tau$ in the domain $D$.  From Eq. \eqref{newg3}, this will follow by proving that $f(\tau, t_{0}(\tau, \rho))$ is differentiable with respect to $\tau$.  This is established by the following theorem whose proof depends on Lemmas \ref{firstlemmatau}, \ref{taulimlemma1}, and \ref{taulimlemma2}.

\begin{theorem}\label{dtaug}
Let $a(t)$ be strongly regular with $\dot{a}(0^+)<\infty$.  Suppose that there is a constant $C>0$ such that

\begin{equation}\label{abound}
\left|\frac{\dddot{{}a}(t)a^2(t)}{\dot{a}^3(t)}\right|\leqslant C,
\end{equation}
for all $t$. Then $g_{\tau\tau}$ is  differentiable with respect to $\tau$ in $D$ and,
\begin{equation}\label{tauderivformula}
\partial_{\tau}g_{\tau\tau}=2\sqrt{-g_{\tau\tau}}\,[\ddot{a}(\tau)\mathbf{f}(\tau,\rho)+\dot{a}(\tau)\partial_\tau \mathbf{f}(\tau,\rho)],
\end{equation}
where

\begin{equation}\label{dtauf}
\partial_\tau \mathbf{f}(\tau,\rho)=\partial_{\tau}f(\tau,t_0)+\partial_{t_0}f(\tau,t_0)\partial_\tau t_0(\tau,\rho),
\end{equation}
 for $(\tau,\rho)\neq(\tau,\rho_{\mathcal{M}_\tau})$, where $\partial_\tau f(\tau,t_0),\partial_{t_0}f(\tau,\rho)$ and $\partial_{\tau}t_0(\tau,t_0)$ are given by Eqs. \eqref{ftauderivative}, \eqref{dt0g} and \eqref{taudp}, respectively. If $(\tau,\rho)=(\tau,\rho_{\mathcal{M}_\tau})$ then,
\begin{equation}
\partial_{\tau}\mathbf{f}(\tau,\rho_{\mathcal{M}_\tau})=\begin{cases}
\tpder{}{\tau}f(\tau,0)+\frac{1}{a(\tau)}\der{\rho_{\mathcal{M}_\tau}}{\tau}(\tau) &\text{ if }\dot{a}(0)=0\\
\tpder{}{\tau}f(\tau,0)+\dot{a}(0^+)\der{\rho_{\mathcal{M}_\tau}}{\tau}(\tau)\int_0^\tau\frac{\ddot{a}(t)}{\dot{a}^2(t)}\frac{dt}{\sqrt{a^2(\tau)-a^2(t)}} &\text{ if } \dot{a}(0^+)>0
\end{cases}
\end{equation}
\end{theorem}

\noindent The lemmas that follow in this section rely on the differentiability of $t_{0}(\tau, \rho)$, on $D$ except possibly where $t_{0}=0$. This follows from the Implicit Function theorem and Eq.\eqref{properAlt2}. 

\begin{lemma}\label{firstlemmatau}
Let $a(t)$ be strongly regular. Suppose that there is a constant $C>0$ such that

\begin{equation}\label{abound}
\left|\frac{\dddot{a}(t)a^2(t)}{\dot{a}^3(t)}\right|\leqslant C.
\end{equation}
Then for any $\tau>0$, $f(\tau,t_0)$ is differentiable with respect to $\tau$, and
\begin{equation}\label{ftauderivative}
\pder{f}{\tau}(\tau,t_0)=\begin{cases}
I_1(\tau,t_0)+I_2(\tau,t_0) &\text{ if }t_0>0\\
2\tpder{}{\tau}f(\tau,0)-I_1(\tau,-t_0)-I_2(\tau,-t_0)& \text{ if } t_0<0
\end{cases}
\end{equation}
\noindent where
\begin{align}\label{I1}
I_1(\tau,t_0)&=-\frac{\dot{a}(\tau)}{a(\tau)}\int_{t_0}^\tau\left[3\frac{\ddot{a}^2(t)a(t)}{\dot{a}^4(t)}-\frac{\dddot{{}a}(t)a(t)}{\dot{a}^3(t)}\right]\left[\frac{\sqrt{a^2(\tau)-a^2(t_0)}}{\sqrt{a^2(\tau)-a^2(t)}}-1\right]dt,
\end{align}

\begin{align}\label{I2}
I_2(\tau,t_0)&=\frac{\dot{a}(\tau)}{a(\tau)}\int_{t_0}^\tau\frac{\ddot{a}(t)}{\dot{a}^2(t)}\left[\frac{a^2(\tau)}{\sqrt{a^2(\tau)-a^2(t)}\sqrt{a^2(\tau)-a^2(t_0)}}-1\right]dt
\end{align}

\noindent and

\begin{equation}\label{partialtauf0} \pder{f}{\tau}(\tau,0)=\frac{\dot{a}(\tau)}{a(\tau)}\int_0^\tau\left[\frac{\ddot{a}(t)}{\dot{a}^2(t)}+\frac{\dddot{{}a}(t)a(t)}{\dot{a}^3(t)}-3\frac{\ddot{a}^2(t)a(t)}{\dot{a}^4(t)}\right]\left[\frac{a(\tau)}{\sqrt{a^2(\tau)-a^2(t)}}-1\right]dt
\end{equation}
\end{lemma}

\begin{lemma}\label{taulimlemma1}
Assume that the conditions of Lemma \ref{firstlemmatau} hold.
Then for any $\tau_0>0$ we have that
\begin{equation}
\lim_{\tau\to\tau_0}\pder{f}{\tau}(\tau,t_0(\tau))=\pder{f}{\tau}(\tau_0,0),
\end{equation}
where $\partial_\tau f(\tau_0,0)$ is given by Eq. \eqref{partialtauf0} and $t_0(\tau)\equiv t_0(\tau,\rho_{\mathcal{M}_{\tau_0}})$.
\end{lemma}

\begin{lemma}\label{taulimlemma2}
Under the assumptions of Lemma \ref{firstlemmatau} and one of the conditions of Lemma \ref{c1lemma2}, we have that 
\begin{equation}
\lim_{\tau\to\tau_0}\pder{f}{t_0}(\tau,t_0(\tau))\pder{t_0}{\tau}(\tau,t_0(\tau))
\end{equation}
exists and is finite, where as in Lemma \ref{taulimlemma1}, $t_0(\tau)\equiv t_0(\tau,\rho_{\mathcal{M}_{\tau_0}})$.
\end{lemma}

\noindent Now that we have established the existence of both partial derivatives of $g_{\tau\tau}$ on the domain $D$ we proceed to show that $g_{\tau\tau}$ is continuously differentiable on $D$.

\begin{theorem}\label{rhoderivg}
Let $a(t)$ be strongly regular and suppose that one of the conditions of Lemma \ref{c1lemma2} hold. Then the partial derivative $\partial_\rho g_{\tau\tau}$ given by Eq. \eqref{drhometric} is continuous on $D$.
\end{theorem}

\begin{theorem}\label{tauderivg}
Suppose that the conditions of Lemma \ref{firstlemmatau} and one of the conditions of Lemma \ref{c1lemma2} hold. Then the partial derivative $\partial_\tau g_{\tau\tau}$ given by Eq. \eqref{tauderivformula} is continuous on $D$.
\end{theorem}

\noindent The following theorem summarizes the results of this section.
\begin{theorem}\label{C1}
Suppose that one of the conditions of Lemma \ref{c1lemma2} hold and that $a(t)$ is strongly regular. Suppose also that there is a constant $C>0$ such that
\begin{equation}\label{abound'}
\left|\frac{\dddot{{}a}(t)a^2(t)}{\dot{a}^3(t)}\right|\leqslant C.
\end{equation}
for all $t$. Then $g_{\tau\tau}(\tau,\rho)$ is continuously differentiable on $D$.
\end{theorem}

\section{$\overline{\mathcal{M}}$ in $1+1$ dimensions}\label{extend}

In this section we define a two dimensional spacetime manifold $\overline{\mathcal{M}}$ that includes pre-big bang events and the $1+1$ dimensional Robertson-Walker spacetime $\mathcal{M}$ as a submanifold.  The extended Robertson-Walker metric on $\overline{\mathcal{M}}$, restricted to the submanifold $\mathcal{M}^{0}$ of cosmological time zero events, is singular in the sense that it is continuously differentiable on $\mathcal{M}^{0}$, but in general not twice differentiable there. We assume throughout this section that the scale factor satisfies the hypotheses of Theorem \ref{C1}.  \\ 

\noindent We begin with the line element for $\mathcal{M}$ in curvature coordinates,   
\begin{equation}\label{frwmetric2d}
ds^2=-dt^2+a^2(t)d\chi^2.
\end{equation}
There is a coordinate singularity in Eq.\eqref{frwmetric} in four spacetime dimensions at $\chi=0$, but this singularity disappears in two spacetime dimensions and $\chi$ may be extended symmetrically to take all real values (for $k=0, -1$).  For the comoving observer at $\chi=0$, the maximal Fermi chart then consists of all $(\tau,\rho)$ with $\tau>0$ and $ |\rho| <\rho_{\mathcal{M}_{\tau}}$ (see Fig. \ref{Milne} and Section \ref{charts}), with the metric given by,
\begin{equation}\label{fermi2}
ds^2=g_{\tau\tau}(\tau,\rho) d\tau^2+d\rho^2,
\end{equation}
where $g_{\tau\tau}$ is defined by Eq.\eqref{newg3}.   From Theorem \ref{C1}, it is easily verified by symmetry that $g_{\tau\tau}$ is continuously differentiable on the set,
\begin{equation}
D'=\{(\tau,\rho): \tau>0, |\rho| <2\rho_{\mathcal{M}_{\tau}}\},
\end{equation}
with $t_{0}=t_{0}(\tau,\rho)$ defined implicitly by a slight modification of Eq.\eqref{properAlt2}:
\begin{equation}\label{properAlt2'}
|\rho|=\int_{t_{0}}^{\tau}\frac{a(t)}{\sqrt{a^{2}(\tau)-a^{2}(t)}}\,dt.
\end{equation}


\noindent In order to define the extended spacetime, $\overline{\mathcal{M}}$, we first extend the maximal Fermi chart of the $\chi=0$ comoving observer.  For that purpose, we use a subset of $D'$.  By Theorem \ref{c0thrm}, $g_{\tau\tau}(\tau, \rho_{\mathcal{M}_{\tau}})<0$, but we have not ruled out the possibility that $g_{\tau\tau}(\tau,\rho)=0$ at points  where 
$t_{0}(\tau,\rho)<0$ in the domain $D'$ (or $D$).  However this is not an essential feature of our construction because the choice we made for the metric tensor for negative cosmological times was arbitrary except at points $(\tau,\rho)$ where $t_{0}(\tau,\rho)<0$ and $t_{0}(\tau,\rho)$ is close to zero.  At such points near the big bang, continuous differentiability of $g_{\tau\tau}$ at $(\tau, \rho_{\mathcal{M}_{\tau}})$ places restrictions on any extension of this function so that $g_{\tau\tau}$ cannot differ greatly at points near the big bang from the definition given in Eq.\eqref{newg3}.   \\

\noindent In accordance with Remark \ref{toobig}, and to eliminate ambiguities in the extension of $g_{\tau\tau}$, let the coordinate chart $D_{2}\subset D'$ be  defined by,
\begin{equation}\label{D2}
D_{2}=\{(\tau,\rho): \tau>0, |\rho|<\rho_{\text max}(\tau)\},
\end{equation}
where,
\begin{equation}\label{D2b}
\rho_{\text max}(\tau)=\inf \{\rho: 0<\rho<2\rho_{\mathcal{M}_{\tau}}\, \text{and}\,\, g_{\tau\tau}(\tau,\rho)=0\},
\end{equation}
provided the infimum exists, with $\rho_{\text{max}}=2\rho_{\mathcal{M}_\tau}$ otherwise.\\ 

\noindent We can now define $\overline{\mathcal{M}}$ as a disjoint union,
\begin{equation}\label{union}
\overline{\mathcal{M}}=\mathcal{M}^{+}\cup\mathcal{M}^{0}\cup\mathcal{M}^{-},
\end{equation}
where the superscripts indicate respectively that cosmological time $t_{0}$ restricted to the set is positive, zero, or negative. Here $\mathcal{M}^{+}= \mathcal{M}$ denotes the original Robertson-Walker universe where cosmological time is postive, and $t_{0}=t_{0}(\tau,\rho)$ is a continuous function of $\tau$ and $\rho$ within the Fermi chart and on $\mathcal{M}^{0}\cup\mathcal{M}^{-}$, all of which are covered by the chart $D_{2}$.  The original smooth charts on $\mathcal{M}^{+}= \mathcal{M}$ together with $D_{2}$ form an atlas on $\overline{\mathcal{M}}$. The metric restricted to the submanifold of $\overline{\mathcal{M}}$ with chart $D_{2}$ is given by Eq. \eqref{fermi2}.  Under the assumptions of Theorem \ref{C1}, the metric of Eq.\eqref{fermi2} is $\mathcal{C}^{1}$ on $\mathcal{M}^{0}$, and smooth on $\overline{\mathcal{M}}\smallsetminus\mathcal{M}^{0}$.

\begin{remark}\label{geodesic} It is easily verified that all connection coefficients from the metric of Eq. \eqref{fermi2} are continuous on $\overline{\mathcal{M}}$, and that the spacelike path $\gamma(\rho)=(\tau_{0}, \rho)$ satisfies the geodesic equations for all $\rho$ with $|\rho|<\rho_{\text max}(\tau_{0})$.  Thus, all spacelike geodesics in $\overline{\mathcal{M}}$ orthogonal to the Fermi observers's worldlike $\beta(\tau)=(\tau, 0)$ pass through spacetime points in the big bang $\mathcal{M}^{0}$ as well as pre-big bang points in $\mathcal{M}^{-}$.
\end{remark}

\noindent  Geometrically, $\rho_{\text max}(\tau_{0})$ may be understood as the first zero of $g_{\tau\tau}(\tau_{0}, \cdot)$ along the spacelike geodesic orthogonal to $\beta$, and starting from, $\beta(\tau_{0})$, as $\rho$ increases.  It follows from Theorem \ref{c0thrm} that $\rho_{\text max}(\tau)>\rho_{\mathcal{M}_{\tau}}$ so that $\mathcal{M}^{-}$ is not empty and necessarily consists of points with negative cosmological times $t_{0}$.\\ 

\noindent The subset $\mathcal{M}^{0}$ defined by $t_{0}(\tau,\rho)=0$ is a $C^{1}$ submanifold parameterized in two connected components by $(\tau, \rho_{\mathcal{M}_{\tau}})$, and  $(\tau, -\rho_{\mathcal{M}_{\tau}})$, for $\tau>0$, where the one-to-one, continuously differentiable function $\rho_{\mathcal{M}_{\tau}}$ is given by Eq.\eqref{radius}. It follows from Remark \ref{regineq} and Eq.\eqref{dradius} that $d\rho_{\mathcal{M}_{\tau}}/d\tau$ is nonvanishing for the scale factors we consider. The charactor of $\mathcal{M}^{0}$ is described by the following theorem.
\begin{theorem}\label{lightlike}
Under the assumptions of Theorem \ref{C1}, 
$\mathcal{M}^{0}$ is lightlike.
\end{theorem}
\begin{proof} A tangent vector to $\mathcal{M}^{0}$ at the point $(\tau, \pm\, \rho_{\mathcal{M}_{\tau}})$ is $u=(1,\pm\, d\rho_{\mathcal{M}_{\tau}}/d\tau)$, (where $+$ is used for the component with positive space coordinates and $-$ for the other component).  From Eq.(66) in \cite{randles} and Eq.(79) and Theorem 8 in \cite{klein13}, it follows that,
\begin{equation}
\sqrt{-g_{\tau\tau}(\tau, \rho_{\mathcal{M}_{\tau}})}=\lim_{\quad\rho\to\rho_{\mathcal{M}_\tau}^{-}}\sqrt{-g_{\tau\tau}(\tau,\rho)}=\lim_{t_{0}\to0^{+}}\frac{\|v_{\mathrm{Fermi}}\|}{\|v_{\mathrm{kin}}\|}=\der{\rho_{\mathcal{M}_\tau}}{\tau}(\tau),
\end{equation}
where the limit in the third term is of the ratio of the Fermi relative speed to the kinematic relative speed of a comoving test particle at the spacetime point uniquely determined by $\tau$ and $t_{0}$ (see \cite{randles,Bolos12,klein13}), and both speeds are relative to the comoving observer $\beta$.  Then,

\begin{equation}
g(u,u)=g_{\tau\tau}(\tau, \rho_{\mathcal{M}_{\tau}})+\left(\der{\rho_{\mathcal{M}_\tau}}{\tau}\right)^{2}=0,
\end{equation}
where $g$ is the metric tensor.
\end{proof}

\section{$\overline{\mathcal{M}}$ in $3+1$ dimensions}\label{3+1}

\noindent  In this section we construct a spacetime $\overline{\mathcal{M}}$ in four spacetime dimensions, analogous to the construction in Section \ref{extend} for the two dimensional case. Analogous to the two dimensional case, here $\overline{\mathcal{M}}$ includes both pre-big bang events and the  $3+1$ dimensional Robertson-Walker spacetime $\mathcal{M}$ as a submanifold, for $k=0$ (see Eq.\eqref{Sk}).  The extended metric tensor on $\overline{\mathcal{M}}$ is smooth except for its restriction to the submanifold $\mathcal{M}^{0}$ of cosmological time zero events. On $\mathcal{M}^{0}$, the extended metric is $\mathcal{C}^{0}$ but also with certain differentiability properties. The big bang, $\mathcal{M}^{0}$, inherits  geometric structure  from $\overline{\mathcal{M}}$, and the dimension of the the cotangent bundle on  $\mathcal{M}^{0}$ is two dimensional.\\  

\noindent  In subsection \ref{angular} we develop extended Fermi polar coordinates and extend the metric of Eqs. \eqref{fermipolar} and \eqref{fermipolarchart}. Because partial derivatives of the metric coefficient $a^{2}(t_{0})S^2_k(\chi_{t_{0}}(\tau))$ in Eq. \eqref{fermipolar} diverge at  cosmological time zero spacetime points, it cannot be extended as a differentiable function of $\tau$ and $\rho$ under the general assumptions that make $g_{\tau\tau}$ continuously differentiable where $t_{0}=0$. Subsection \ref{7.2fermi} develops the full (Cartesian) Fermi coordinates $\tau, x, y, z$ of Eq.\eqref{fermimetric} and finishes the construction of $\overline{\mathcal{M}}$. 

\subsection{Angular Coordinates}\label{angular}

\noindent Here we define an extension, 
\begin{equation}\label{extendpolar}
ds^2=g_{\tau\tau} d\tau^2+d\rho^2 + g_{\theta\theta}\,d\theta^{2}+g_{\phi\phi}\,d\phi^{2},
\end{equation}
of the Eq.\eqref{fermipolar} in Fermi polar form.  A new chart $D_{\mathrm{polar}}\supset \mathcal{U}_{\mathrm{polar}}$ (see Eq \eqref{fermipolarchart}) is defined by,
\begin{equation}\label{fermipolarchart+}
D_{\mathrm{polar}}=\{(\tau,\rho,\theta,\phi): \tau>0, 0<\rho<\rho_{\text max}(\tau), \theta\in I_{\pi},\, \phi\in I_{2\pi}\},
\end{equation}
where the notation is the same as in Eq. \eqref{D2b}, and as before $I_{\pi}$ and $I_{2\pi}$ may be chosen to be any open intervals of length $\pi$ and $2\pi$ respectively.  As before, we assume that the scale factor $a(t)$ is extended as an even function of $t$. It follows from Theorem \ref{C1} that $g_{\tau\tau}$ is continuously differentiable on $D_{\mathrm{polar}}$ under the hypotheses of that theorem.  As in Section \ref{extend}, we define $g_{\rho\rho}\equiv1$ and define $g_{\tau\tau}$ by Eq. \eqref{newg3} on $D_{\mathrm{polar}}$. For $t_{0}(\tau,\rho)\neq0$, let

\begin{equation}\label{chitermk'}
g_{\theta\theta}(\tau,\rho)=\bar{g}_{\theta\theta}(\tau,t_0(\tau,\rho))=a^{2}(t_{0})S^2_k(\chi_{|t_{0}|}(\tau)),
\end{equation}
and
\begin{equation}\label{same}
g_{\phi\phi}(\tau,\rho,\theta)=g_{\theta\theta}(\tau,\rho)\sin^{2}\theta.
\end{equation}
In what follows, we shall define $g_{\theta\theta}$ and $g_{\phi\phi}$, at points $(\tau,\rho,\theta,\phi)\in D_{\mathrm{polar}}$ where $t_{0}(\tau,\rho)=0$, i.e., where $\rho=\rho_{\mathcal{M}_\tau}$, by the limiting values of those functions as $t_{0}\to0$.\\

\noindent To define $g_{\theta\theta}$ at points where $t_{0}=0$, we first consider the case $k=1$.  Since $S_{k=1}$ is bounded, the right side of Eq. \eqref{chitermk'} converges to zero as $t_{0}\rightarrow0$, so $g_{\theta\theta}$ and $g_{\phi\phi}$ can both be defined to take the value zero at such points. For the other cases, $k=0,-1$, we first observe that from Eq.\eqref{key2} and Lemma \ref{intbound}, for a regular scale factor $a(t)$,
\begin{equation}\label{chibound}
\int_{t_{0}}^{\tau}\frac{1}{a(t)}dt\leq\chi_{t_{0}}(\tau)< \,\frac{1}{\dot{a}(\tau)}+\int_{t_{0}}^{\tau}\frac{1}{a(t)}dt
\end{equation}
for $\tau>t_{0}\geq0$.  In both cases $S_{k}$ is an increasing function, so combining \eqref{chitermk'} and \eqref{chibound}, gives,
\begin{equation}\label{metricbound}
g_{\theta\theta}(\tau,\rho)=a^{2}(t_{0})S^2_k(\chi_{|t_{0}|}(\tau))\leq a^{2}(t_{0})S^2_k\left(\frac{1}{\dot{a}(\tau)}+ \int_{|t_{0}|}^{\tau}\frac{1}{a(t)}dt\right). 
\end{equation}
By Theorem \ref{inflation}, the assumptions we make in Theorem  \ref{C1} and Lemma  \ref{c1lemma2} are inconsistent with the existence of finite particle horizons in the cosmologies we consider. Nevertheless we point out that if the scale factor $a(t)$ (in violation with those hypotheses) does give rise to finite particle horizons, i.e., 
\begin{equation}\label{particlehor}
\int_{0}^{\tau}\frac{1}{a(t)}dt<\infty,
\end{equation}
then we have the following result.
\begin{theorem} \label{finitepart} For $k=-1,0,1$ and a regular scale factor $a(t)$ with finite particle horizon, i.e., satisfying Eq. \eqref{particlehor} for $\tau>0$, $g_{\theta\theta}(\tau,\rho)$ and $g_{\phi\phi}(\tau,\rho,\phi)$ are continuous on $D_{\mathrm{polar}}$ and vanish at $(\tau, \rho_{\mathcal{M}_\tau})$.
\end{theorem}
\begin{proof} It is necessary to show continuity only at points of the form $(\tau, \rho_{\mathcal{M}_\tau}, \theta,\phi)$.  Since $t_{0}(\tau,\rho)=0$ if and only if $\rho = \rho_{\mathcal{M}_\tau}$ and $t_{0}(\tau,\rho)$ is continuous on $D$ by Lemma \ref{t0cont}, it suffices to show that $\bar{g}_{\theta\theta}$ is continuous at $(\tau_{0},0)$ for any $\tau_{0}>0$.  Using Eq.\eqref{metricbound}, we have,

\begin{equation}\label{finitepartcont}
\lim_{(\tau,t_{0})\to(\tau_0,0)}\bar{g}_{\theta\theta}(\tau,t_0)=\lim_{(\tau,t_{0})\to(\tau_0,0)}a^{2}(t_{0})S_k^{2}(\chi_{t_{0}}(\tau))= 0
\end{equation}
\end{proof}
\noindent The following remark shows that the case $k=-1$ is problematic for inflationary cosmologies. 

\begin{remark}\label{Mink1} In the Milne Universe $k=-1$ and Remark \ref{Mink2} shows that $g_{\theta\theta}$ and $g_{\phi\phi}$ have obvious smooth extensions to $D_{\mathrm{polar}}$. If $a(t)=t^{\alpha}$ and $\alpha<1$, then particle horizons are finite and Theorem \ref{finitepart} applies. However, for $k=-1$ with inflationary power law scale factors of the form $a(t)=t^{\alpha}$ with $\alpha>1$, it is readily seen that 
\begin{equation}
\lim_{t_{0}\rightarrow0^{+}}a(t_{0})S_k(\chi_{t_{0}}(\tau))=\lim_{t_{0}\rightarrow0^{+}}a(t_{0})\sinh(\chi_{t_{0}}(\tau))= \infty.
\end{equation}
Therefore continuous extensions of $g_{\theta\theta}(\tau,\rho)$ and $g_{\phi\phi}(\tau,\rho,\phi)$ to $D_{\mathrm{polar}}$ are not possible for these cosmologies.  
\end{remark}

\noindent Specializing to the case $k=0$ and regular scale factors with infinite particle horizons, we have, using Eq.\eqref{chibound} and L'H\^{o}pital's rule,

\begin{equation}\label{H(0+)}
\lim_{t_{0}\rightarrow0}a(t_{0})\chi_{|t_{0}|}(\tau) = \lim_{t_{0}\rightarrow0}\frac{a(t_{0})}{\dot{a}(t_{0})}=\lim_{t_{0}\rightarrow0}\frac{1}{H(t_{0})}= \frac{1}{H(0^{+})} ,
\end{equation}
which exists because the Hubble parameter $H(t_{0})$ is a decreasing function for regular scale factors.

\begin{theorem}\label{gthetacontin}
Let $k=0$ and $a(t)$ be a regular scale factor. Then $g_{\theta\theta}(\tau,\rho)$ and $g_{\phi\phi}(\tau,\rho,\phi)$ are continuous on $D_{\mathrm{polar}}$.
\end{theorem}

\begin{proof}
It is sufficient to prove that $g_{\theta\theta}$ is continuous at points of the form $(\tau, \rho_{\mathcal{M}_\tau})$. For the case of finite particle horizons, the result follows by Theorem \ref{finitepart}. For the case of infinite particle horizons, we define a new function of two independent variables $\tau$ and $x$, by,
\begin{equation}\label{phidef}
F_{4}(\tau,x)\equiv a(x)\int_{|x|}^\tau\frac{1}{a(t)}\frac{a(\tau)}{\sqrt{a^2(\tau)-a^2(t)}}dt\equiv a(x)\int_{|x|}^\tau h_{4}(\tau,t)dt,
\end{equation}
with the restriction $\tau>x>0$.  In light of Eq. \eqref{H(0+)}, we define $F_{4}(\tau,0)= 1/H(0^{+})$.   Since $t_{0}(\tau_{0},\rho_{\mathcal{M}_{\tau_{0}}})=0$, our plan of proof is first to show that $F_{4}(\tau,x)$ is continuous at any point of the form $(\tau_{0},0)$. Then using Lemma \ref{t0cont},  the composition $F_{4}(\tau, t_{0}(\tau,\rho))$ must be continuous at any point of the form $(\tau_{0},\rho_{\mathcal{M}_{\tau_{0}}})$ and and therefore $g_{\theta\theta}(\tau,\rho)=F_{4}^{2}(\tau, t_{0}(\tau,\rho))$ must be continuous at any point of the form $(\tau_{0},\rho_{\mathcal{M}_{\tau_{0}}})$. \\

\noindent For $\tau>\tau_0>x>0$,
\begin{multline}\label{cont1}
|F_{4}(\tau,x)-F_{4}(\tau_0,0)|\leqslant |F_{4}(\tau,x)-F_{4}(\tau_0,x)|+|F_{4}(\tau_0,x)-1/{H(0^+)}|\\
\leqslant a(x)\left|\int_{\tau_0}^\tau h_{4}(\tau,t)dt\right|+a(x)\int_x^{\tau_0}|h_{4}(\tau,t)-h_{4}(\tau_0,t)|dt+|F_{4}(\tau_0,x)-1/H(0^+)|
\end{multline}
The third term in Eq. \eqref{cont1} can be made arbitrarily small for all $x$ sufficiently close to 0 by Eq. \eqref{H(0+)}, while the first term can be made small for all $\tau$ sufficiently close to $\tau_0$ by Lemma \ref{continuitylemma}b. For the middle term, choose any $0<\delta<\tau_0$. We can assume that $x<\tau_0-\delta$. Then,
\begin{multline}\label{inequal1}
a(x)\int_x^{\tau_0}|h_{4}(\tau,t)-h_{4}(\tau_0,t)|dt\\
=a(x)\int_{\tau_0-\delta}^{\tau_0}|h_{4}(\tau,t)-h_{4}(\tau_{0},t)|dt+a(x)\int_x^{\tau_0-\delta}|h_{4}(\tau,t)-h_{4}(\tau_0,t)|dt
\end{multline}
The first integral is bounded by the  integrability of $h_{4}(\tau_0,t)$ and $h_{4}(\tau,t)$, so it can be made small for small enough $a(x)$, which is achieved by choosing $x$ sufficiently close to 0. For the second term, note that $a(t)h_{4}(\tau,t)$ is uniformly continuous in $t$ and $\tau$ for $t\in[0,\tau_0-\delta]$. So, for any $\epsilon>0$ we can choose $\tau$ close enough to $\tau_0$ so that,
\begin{equation}
a(x)\int_x^{\tau_0-\delta}\frac{1}{a(t)}|a(t)h_{4}(\tau,t)-a(t)h_{4}(\tau_0,t)|dt\leqslant \epsilon a(x)\int_x^{\tau_0-\delta}\frac{1}{a(t)}dt.
\end{equation}
Since $\lim_{x\to 0}a(x)\int_x^{\tau_0-\delta}\frac{1}{a(t)}dt=1/H(0+)$, this term is bounded for $x$ close to 0. Therefore we can make the entire second term in Eq. \eqref{inequal1} small by choosing $(\tau,x)$ sufficiently close to $(\tau_{0},0)$. The $\tau<\tau_0$ case is simliar.
\end{proof}

\noindent The following corollary shows that if $k=0$ and the scale factor $a(t)$ is analytic at $t_{0}=0$ and regular, then $g_{\theta\theta}(\tau, \rho_{\mathcal{M}_\tau}) = g_{\phi\phi}(\tau, \rho_{\mathcal{M}_\tau},\theta)=0$. 

\begin{corollary} \label{anglelimit} If $k=0$ and $a(t)$ is regular and an even function on $\mathbb{R}$, with either $\infty>\dot{a}(0^{+})>0$ or $\dot{a}(0)=0$ but the $n$th derivative of $a(t)$ at $t=0$ for some $n$ does not vanish, then $g_{\theta\theta}(\tau,\rho)$ and $g_{\phi\phi}(\tau,\rho,\phi)$ are continuous on $D_{\mathrm{polar}}$ and vanish at $(\tau, \rho_{\mathcal{M}_\tau})$ for any $\tau>0$.
\end{corollary}
\begin{proof} Continuity follows from Theorem \ref{gthetacontin} and if the particle horizon is finite, $g_{\theta\theta}(\tau,\rho)$ and $g_{\phi\phi}(\tau,\rho,\phi)$ vanish at $(\tau, \rho_{\mathcal{M}_\tau})$ by Theorem \ref{finitepart}. If the particle horizon is infinite and $\dot{a}(0^{+})>0$, the result follows from  Eq.\eqref{H(0+)} and the assumption that $a(0)=0$.  Alternatively, if $\dot{a}(0)=0$ but $a^{(n)}(0)\neq0$ for some $n$, then for the smallest such $n$, we must have
\begin{equation}
a(0)=\dot{a}(0)=\ddot{a}(0)=\cdots=a^{(n-1)}(0)=0,
\end{equation}
but $a^{(n)}(0)\neq0$. Repeated application of L'H\^{o}pital's rule then gives,
\begin{equation}
\lim_{t_{0}\to0^{+}}a(t_{0})\chi_{t_{0}}(\tau) = \lim_{t_{0}\rightarrow0}\frac{a^{(n-1)}(t_{0})}{a^{(n)}(t_{0})}=0.
\end{equation}
So,
\begin{equation}\label{finitepartcont'}
g_{\theta\theta}(\tau, \rho_{\mathcal{M}_\tau})=\lim_{(\tau,t_{0})\to(\tau_0,0)}a^{2}(t_{0})S_k^{2}(\chi_{|t_{0}|}(\tau))=\lim_{t_{0}\rightarrow0^{+}}a^{2}(t_{0})\chi_{t_{0}}^{2}(\tau)= 0.
\end{equation}
Eq.\eqref{same} finishes the proof.
\end{proof}

\noindent We collect results from this subsection and Theorem \ref{C1} in the following theorem.

\begin{theorem}\label{C0}
Suppose that one of the conditions of Lemma \ref{c1lemma2} hold and that $a(t)$ is strongly regular. Suppose also that there is a constant $C>0$ such that
\begin{equation}\label{abound'}
\left|\frac{\dddot{a}(t)a^2(t)}{\dot{a}^3(t)}\right|\leqslant C.
\end{equation}
for all $t$. Then $g_{\tau\tau}(\tau,\rho)$ is continuously differentiable on $D_{\mathrm{polar}}$ and if $k=0$,  $g_{\theta\theta}(\tau,\rho)$ and $g_{\phi\phi}(\tau,\rho,\phi)$  are continuous on $D_{\mathrm{polar}}$.
\end{theorem}

\subsection{Fermi Coordinates}\label{7.2fermi}

%
\noindent  We begin with a definition of the key chart for the extended spacetime $\overline{\mathcal{M}}$ in $3+1$ spacetime coordinates.  Let $D_{\mathrm{Fermi}}\supset\mathcal{U}_{\mathrm{Fermi}}$ of Eq.\eqref{UFermi} be defined by,  
\begin{equation}\label{UFermi2}
D_{\mathrm{Fermi}} = \left\{(\tau,x,y,z): \tau >0 \text{ and } \sqrt{x^2+y^2+z^2} < \rho_{\text max}(\tau)\right\},
\end{equation}
where $\rho_{\text max}(\tau)$ is given by Eq.\eqref{D2b}. In light of Eq.\eqref{chitermk'} and Theorem \ref{gthetacontin}, $\lambda_{k}(\tau,\rho)$ in Eq.\eqref{lambda2} may be extended as a continuous function to $D_{\mathrm{Fermi}}$ by the formula,
\begin{equation}\label{lambda2'}
\lambda_{k}(\tau,\rho) =\frac{g_{\theta\theta}(\tau,\rho)-\rho^{2}}{\rho^{4}},
\end{equation}
for a regular scale factor $a(t)$.  In particular,
\begin{equation}
\begin{split}
\lambda_{k}(\tau,\rho_{\mathcal{M}_\tau})&=\lim_{\quad\rho\to\rho_{\mathcal{M}_\tau}^{-}}\lambda_{k}(\tau,\rho)\\
&=\lim_{\quad\rho\to\rho_{\mathcal{M}_\tau}^{-}}\frac{a^{2}(t_{0})S^2_k(\chi_{t_{0}}(\tau))-\rho^{2}}{\rho^{4}}
=\frac{1/H^2(0^+)-\rho^2_{\mathcal{M}_\tau}}{\rho_{\mathcal{M}_\tau}^{4}}.
\end{split}
\end{equation}

\noindent We note that if the hypotheses of Corollary \ref{anglelimit} are satisfied, then $\lambda_{k}(\tau,\rho_{\mathcal{M}_\tau})=-1/\rho_{\mathcal{M}_\tau}$.\\  

\noindent In analogy with Eq.\eqref{union} for the two dimensional case, we can now define the four dimensional spacetime $\overline{\mathcal{M}}$ as a disjoint union,
\begin{equation}\label{union4}
\overline{\mathcal{M}}=\mathcal{M}^{+}\cup\mathcal{M}^{0}\cup\mathcal{M}^{-},
\end{equation}
where, as before, the superscripts indicate respectively that cosmological time $t_{0}$ restricted to the set is positive, zero, or negative. Here $\mathcal{M}^{+}= \mathcal{M}$ denotes the original four dimensional Robertson-Walker universe where cosmological time is postive, and $t_{0}=t_{0}(\tau,\rho)$ is a continuous function of $\tau$ and $\rho$ within the Fermi chart and on $\mathcal{M}^{0}\cup\mathcal{M}^{-}$, all of which are covered by the chart $D_{\mathrm{Fermi}}$.  The original smooth charts on $\mathcal{M}^{+}= \mathcal{M}$ together with $D_{\mathrm{Fermi}}$ form an atlas on $\overline{\mathcal{M}}$. The metric restricted to the submanifold of $\overline{\mathcal{M}}$ with chart $D_{\mathrm{Fermi}}$ is given by,  
\begin{equation}
\begin{split}\label{fermimetric+}
ds^2=&\,g_{\tau\tau}(\tau,\rho)\, d\tau^2+dx^2 +dy^2+dz^2\\ 
+&\lambda_{k}(\tau,\rho)\big[(y^2+z^2)dx^2+(x^2+z^2)dy^2+(x^2+y^2)dz^2\\
-&xy(dxdy+dydx)-xz(dxdz+dzdx)-yz(dydz+dzdy)\big],
\end{split}
\end{equation} 
Under the assumptions of Theorem \ref{C1}, the metric of Eq.\eqref{fermimetric+} is smooth on $\overline{\mathcal{M}}\smallsetminus\mathcal{M}^{0}$.  The next theorem summarizes the results of this subsection.

\begin{theorem}\label{C02}
Suppose that one of the conditions of Lemma \ref{c1lemma2} hold and that the scale factor $a(t)$ is strongly regular. Suppose also that there is a constant $C>0$ such that
\begin{equation}\label{abound'}
\left|\frac{\dddot{a}(t)a^2(t)}{\dot{a}^3(t)}\right|\leqslant C,
\end{equation}
for all $t$. Then $g_{\tau\tau}(\tau,\rho)$ is continuously differentiable on $D_{\mathrm{Fermi}}$ and $\lambda_{k=0}(\tau,\rho)$ is continuous on $D_{\mathrm{Fermi}}$.
\end{theorem}
\begin{remark}\label{qadot}
A sufficient condition for the bound Eq.\eqref{abound'} under the hypotheses of Theorems \ref{C0} and \ref{C02} is the existence of a constant $C$ such that, 

\begin{equation}\label{abound''}
\left|\frac{\dddot{a}(t)\dot{a}(t)}{\ddot{a}^2(t)}\right|\leqslant \frac{C}{K^{2}},
\end{equation}
so that, for example, if $C=K^{2}$, then $\dot{a}(t)$ itself regarded as a scale factor would be regular according to Definition \ref{regular}.  This follows from the implication,
\begin{equation}
\left|\frac{a(t)\ddot{a}(t)}{\dot{a}(t)^{2}}\right|\leq K \implies \left|\frac{a(t)}{\dot{a}(t)}\right|\leq K\left|\frac{\dot{a}(t)}{\ddot{a}(t)}\right|.
\end{equation}
\end{remark}

\noindent The following corollary, established by direct calculation, shows consistency with the polar and cartesian forms of the metric extended to $D_{\mathrm{Fermi}}$.

\begin{corollary} \label{rank}
Under the hypotheses of Corollary \ref{anglelimit}, the Fermi metric of Eq.\eqref{fermimetric}  expressed as a $4$ by $4$ matrix, 
\[ \left( \begin{array}{cccc}
g_{\tau\tau} & 0 & 0 & 0 \\
0 & 1+\lambda_{k}(y^{2}+z^{2}) & -\lambda_{k}xy& -\lambda_{k}xz\\
0 & -\lambda_{k}xy & 1+\lambda_{k}(x^{2}+z^{2}) & -\lambda_{k}yz\\
0 & -\lambda_{k}xz & -\lambda_{k}yz & 1+\lambda_{k}(x^{2}+y^{2})
 \end{array} \right)\] 
and evaluated at $(\tau,\rho_{\mathcal{M}_\tau})$ has rank $2$ for all $\tau>0$.  Thus the cotangent space at each point in $\mathcal{M}^{0}$ is two dimensional.
\end{corollary}

\begin{remark}\label{geodesic2} The submanifold obtained by assigning fixed values  $\theta_{0}$ and $\phi_{0}$ to the angular coordinates in the chart $D_{\mathrm{polar}}$ for $\overline{\mathcal{M}}$ is the two dimensional spacetime analyzed in Section \ref{extend}. From Remark \ref{geodesic}, it follows that the spacelike path $\gamma(\rho)=(\tau_{0}, \rho, \theta_{0}, \phi_{0})$ with $|\rho|<\rho_{\text max}(\tau_{0})$ is geodesic in the submanifold, intersects the big bang $\mathcal{M}^{0}$, and reaches pre-big bang points in $\mathcal{M}^{-}$.
\end{remark}

\noindent The next theorem collects assumptions from Theorems \ref{gcontin} and \ref{gthetacontin} needed for a continuous (not necessarily differentiable) extension of the metric to $\mathcal{M}^{0}$.

\begin{theorem}\label{justcontin}
Let $k=0$ and let $a(t)$ be a strongly regular scale factor. Then the metric coefficients of Eq.\eqref{fermimetric+} are continuous on $D_{\mathrm{Fermi}}$, and the metric coefficients for the polar form, Eq.\eqref{extendpolar}, are continuous on $D_{\mathrm{polar}}$.
\end{theorem}
\noindent We note that the metric is as smooth as $a(t)$ off of the big bang $\mathcal{M}^{0}$.

\section{Examples}\label{examples}

In this section we give examples of Robertson-Walker cosmologies satisfying the conditions of  Theorems \ref{C1}, \ref{C0} and \ref{C02}. We begin with power law cosmologies, i.e., those with scale factors of the form $a(t)=t^\alpha$ with $\alpha>0$. These cosmologies include the radiation-dominated and matter-dominated universes, and models for dark energy (see \cite{power}).\\

\noindent For the power law scale factor $a(t)=t^\alpha$, Eq. \eqref{newg} gives,
\begin{equation}
g_{\tau\tau}(\tau,\rho)=-\left[\left(\frac{\tau}{t_0}\right)^{\alpha-1}-\sqrt{1-\left(\frac{t_0}{\tau}\right)^{2\alpha}}\left(\left(\frac{\tau}{t_0}\right)^{\alpha-1}{}_2F_1\left(\frac{1}{2},\frac{1-\alpha}{2\alpha};\frac{1+\alpha}{2\alpha};\left(\frac{t_0}{\tau}\right)^{2\alpha}\right)-C_\alpha\right)\right]^2,
\end{equation} 

\noindent where $t_0$ is given implicitly as a function of $\tau$ and $\rho$ by Eq. \eqref{properAlt}, ${}_2F_1(\cdot,\cdot;\cdot;\cdot)$ is the Gauss hypergeometric function, and (see \cite{Bolos12,sam}),

\begin{equation}\label{Ca}
C_\alpha=\frac{\rho_{\mathcal{M}_{\tau}}}{\tau}=\frac{\sqrt{\pi}\,\Gamma(\frac{1+\alpha}{2\alpha})}{\Gamma(\frac{1}{2\alpha})}.
\end{equation}
With Theorem \ref{gcontin}, it follows that,
\begin{equation}
\lim_{\quad\rho\rightarrow\rho_{\mathcal{M}_{\tau}}^{-}}g_{\tau\tau}(\tau,\rho)=g_{\tau\tau}(\tau, \rho_{\mathcal{M}_{\tau}})=-C_\alpha^2.
\end{equation}

\noindent Also by Eq. \eqref{properAlt},

\begin{equation}
\rho=\tau\left[C_\alpha-\frac{1}{1+\alpha}\left(\frac{t_0}{\tau}\right)^{1+\alpha}{}_2F_1\left(\frac{1}{2},\frac{1+\alpha}{2\alpha};\frac{1+3\alpha}{2\alpha};\left(\frac{t_0}{\tau}\right)^{2\alpha}\right)\right].
\end{equation}

\noindent Implicit differentiation or Eq.\eqref{taudp} gives,

\begin{equation}
\pder{t_0}{\tau}=\frac{\sqrt{\tau^{2\alpha}-t_0^{2\alpha}}}{t_0^\alpha}\frac{\rho}{\tau}+\frac{t_0}{\tau},
\end{equation}
and further calculations show,

\begin{equation}
\partial_\rho g_{\tau\tau}(\tau,\rho)=\frac{2\alpha}{\tau}\sqrt{-g_{\tau\tau}(\tau,t_0)}\left({}_2F_1\left(\frac{1}{2},\frac{1-\alpha}{2\alpha};\frac{1+\alpha}{2\alpha};\left(\frac{t_0}{\tau}\right)^{2\alpha}\right)-\left(\frac{t_0}{\tau}\right)^{\alpha-1}C_\alpha\right),
\end{equation}

\noindent and,

\begin{align}
\partial_{\tau}g_{\tau\tau}(\tau,\rho)&=-\frac{2\alpha}{\tau}\sqrt{-g_{\tau\tau}(\tau,t_0)}\left({}_2F_1\left(\frac{1}{2},\frac{1-\alpha}{2\alpha};\frac{1+\alpha}{2\alpha};\left(\frac{t_0}{\tau}\right)^{2\alpha}\right)-\left(\frac{t_0}{\tau}\right)^{\alpha-1}C_\alpha\right)\times\nonumber\\
&\left(C_\alpha-\frac{1}{\alpha+1}\left(\frac{t_0}{\tau}\right)^{1+\alpha}{}_2F_1\left(\frac{1}{2},\frac{1+\alpha}{2\alpha};\frac{1+3\alpha}{2\alpha};\left(\frac{t_0}{\tau}\right)^{2\alpha}\right)\right)\nonumber\\
&=-\frac{\rho}{\tau}\partial_\rho g_{\tau\tau}(\tau,\rho).
\end{align}

\noindent Hence,

\begin{equation}\label{limit1}
\partial_\rho g_{\tau\tau}(\tau, \rho_{\mathcal{M}_{\tau}})=\lim_{\quad\rho\rightarrow\rho_{\mathcal{M}_{\tau}}^{-}}\partial_\rho g_{\tau\tau}(\tau,\rho)=
\begin{cases}
\frac{2\alpha C_\alpha}{\tau}&\text{if }\alpha>1\\
-\infty&\text{if }0<\alpha<1
\end{cases}
\end{equation}

\noindent and,

\begin{equation}\label{limit2}
\partial_\tau g_{\tau\tau}(\tau, \rho_{\mathcal{M}_{\tau}})=\lim_{\quad\rho\rightarrow\rho_{\mathcal{M}_{\tau}}^{-}}\partial_{\tau}g_{\tau\tau}(\tau,\rho)=
\begin{cases}
\frac{-2\alpha C_\alpha^2}{\tau}&\text{if }\alpha>1\\
\infty&\text{if }0<\alpha<1.
\end{cases}
\end{equation}
If $\alpha<1$, then $\dot{a}(0^+)=\infty$ and $a(t)=t^{\alpha}$ is not inflationary near $t=0$ (nor at any time) and so $a(t)$ fails to satisfy the hypotheses of Theorems \ref{C1} or \ref{C02}.  For the Milne case (see Remarks \ref{Mink2} and \ref{Mink1}), $\alpha=1$ and taking into account that $\dot{a}(0)=1>0$, the hypotheses of Theorem \ref{C1} are satisfied, so the implied properties of Fig.\ref{Milne} given in the introduction follow from that theorem.\\  

\noindent If $\alpha>1$, then the scale factor $a(t)=t^{\alpha}$ is inflationary near $t=0$ (and for all $t$), and satisfies the hypotheses of Theorems \ref{C1} and \ref{C02} for $k=0$, so the conclusions hold for the cosmologies that are inflationary at the big bang, as shown by Eqs. \ref{limit1} and \ref{limit2}.\\

\begin{figure}[!h]
\begin{center}
\includegraphics[trim=0mm 85mm 4mm 85mm,clip,scale=0.52]{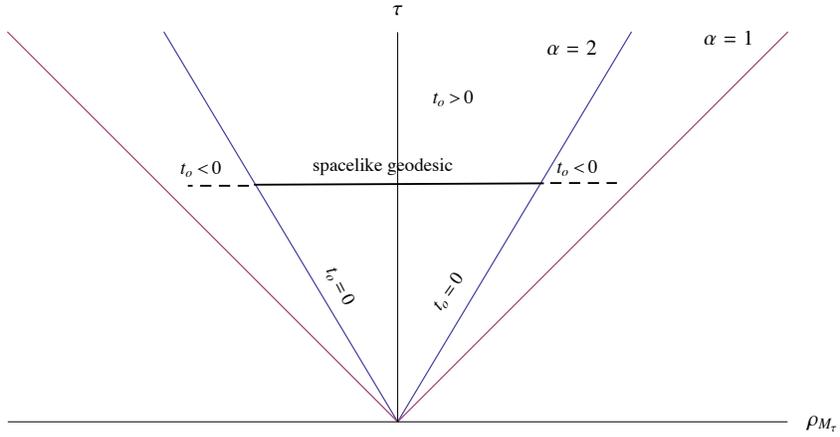}
\end{center}
\caption{A portion of the extended Robertson-Walker cosmology $\overline{\mathcal{M}}$ with scale factor $a(t)=t^{2}$ (i.e. $\alpha=2$).  The Milne Universe (see Fig.\ref{Milne}) with $\alpha=1$ is included only for comparison of Fermi radii, $\rho_{\mathcal{M}_{\tau}}$.  The comoving observer's worldline is the vertical line $\rho=0$ in the center. The dashed portion of the horizontal line extends the spacelike geodesic beyond the $\alpha=2$ universe $\mathcal{M}^{+}=\mathcal{M}$ through $t_{0}=0$ to include points with negative cosmological times. }  
\label{alpha=2}
\end{figure}

\noindent Figure \ref{alpha=2} depicts part of the extension $\overline{\mathcal{M}}$ of the Robertson-Walker cosmology with power law scale factor $a(t)=t^{2}$, i.e., $\alpha=2$. From Eq. \eqref{Ca} the boundary $\mathcal{M}^{0}$ of the Fermi chart in $(\tau,\rho)$ coordinates is,
\begin{equation}
\rho_{\mathcal{M}_{\tau}}=\frac{\sqrt{\pi}\,\Gamma(\frac{1+\alpha}{2\alpha})}{\Gamma(\frac{1}{2\alpha})}\tau=\frac{\sqrt{\pi}\,\Gamma(\frac{3}{4})}{\Gamma(\frac{1}{4})}\tau\approx0.6\tau.
\end{equation}
The Milne Universe ($\alpha=1$) in $(\tau,\rho)$ coordinates, whose boundary satisfies $\rho_{\mathcal{M}_{\tau}}=\tau$, is superimposed for comparison only.
\begin{remark} It can be shown for two dimensional power law cosmologies with $\alpha\geq1$, using Eqs. \ref{limit1} and \ref{limit2}, that $\mathcal{M}^{0}$ parameterized as $(\tau, \pm\rho_{\mathcal{M}_{\tau}})$ is both geodesic with affine parameter $\tau$ and lightlike (see Theorem \ref{lightlike}).
\end{remark}

\noindent For another class of examples, consider a Robertson-Walker cosmology with cosmological constant $\Lambda>0$, curvature parameter $k=0$, and an equation of state for a perfect fluid of the form, $p=(\gamma-1)\rho$, where for our present purposes only, $p$ is pressure, $\rho$ is energy density, and $1\leq\gamma\leq2$ is a constant. For a universe with positive cosmological constant and with matter alone, $\gamma=1$.  For radiation but no matter, $\gamma=4/3$.  The scale factor is then determined \cite{GP} by the Einstein field equations as,
\begin{equation}\label{lambdamatter}
a(t)= A\left[\sinh\left(\frac{3}{2}\sqrt{\frac{\Lambda}{3}}\,\gamma \,t\right)\right]^{2/3\gamma},
\end{equation}
for a constant $A$.  It can be shown that the scale factors given by Eq.\eqref{lambdamatter} satisfy the hypotheses of Theorems \ref{C1} and \ref{C02} for $k=0$, except that $a(t)$ is not inflationary near $t=0$. However, if $\gamma<2/3$,  in which case the pressure of the perfect fluid is negative, then all the hypotheses are satisfied.  For a scale factor that takes the form Eq.\eqref{lambdamatter} with $\gamma<2/3$ close to the big bang at $t_{0}=0$, later evolving into a scale factor with the same form but with $1\leq\gamma\leq2$, the conditions of Theorems \ref{C1} and \ref{C02} can be satisfied, and so the cosmologies can be extended as indicated by those theorems.

\section{Conclusions}\label{conclusions}

\noindent We have constructed geometric extensions of Robertson-Walker cosmologies by extending all spacelike geodesics orthogonal to a comoving observer through and beyond the big bang. Such extensions are possible because, under general conditions, spacelike geodesics with different initial points along a comoving observer's worldline do not intersect \cite{randles,klein13}. Although the extension $\mathcal{M}\to\overline{\mathcal{M}}$  is geometric and coordinate independent, Fermi coordinates are particularly well suited for the construction because the spacelike geodesics are coordinate curves in Fermi coordinates.\\

\noindent In the construction of $\overline{\mathcal{M}}$, the geometry of $\mathcal{M}^{-}$, the submanifold of pre-big bang points, is \textit{a priori} arbitrary, but it is natural to restrict the geometry by imposing as much regularity on the Fermi metric across the big bang $\mathcal{M}^{0}$ as possible, and to examine the restrictions on the scale factor required to achieve that regularity.  The metric in the extended Fermi coordinates is given by Eq.\eqref{fermimetric+}, and in polar form, may be expressed as,
\begin{equation}\label{fermipolarlast}
ds^2=g_{\tau\tau}(\tau,\rho) d\tau^2+d\rho^2 + g_{\theta\theta}(\tau,\rho) d\theta^{2}+g_{\phi\phi}(\tau,\rho,\theta) d\phi^{2}.
\end{equation}
Under the assumptions we make for $a(t)$,  $g_{\tau\tau}$ is continuously differentiable on $\mathcal{M}^{0}$, and with the curvature parameter $k=0$, $g_{\theta\theta}$ and $g_{\phi\phi}$ are continuous there (and the metric is otherwise smooth on $\overline{\mathcal{M}}$). In general greater smoothness is not possible, but even this degree of regularity necessarily constrains the geometry of $\mathcal{M}^{-}$. For a scale factor with nonvanishing $n$th derivative at $t=0$, the coefficients $g_{\theta\theta}$ and $g_{\phi\phi}$ necessarily vanish on $\mathcal{M}^{0}$ by Corollary \ref{anglelimit} which gives the big bang a two dimensional structure.  \\

\noindent The restrictions on the scale factor $a(t)$ needed to accomplish this regularity across the big bang are broadly consistent with observations.   In addition to spatial flatness, the main restriction is that $a(t)$ is increasing (so that the universe is expanding) and that $\ddot{a}(t)>0$ in a neighborhood of the big bang, i.e., inflation occurred right after the big bang.\footnote{The condition $\infty>\dot{a}(0^{+})>0$ may be substituted for inflation near the big bang.}\\  

\noindent Additional restrictions are also needed.  Regularity of the scale factor (see Definition \ref{regular}) may be understood in terms of the Hubble parameter, $H(t)=\dot{a}(t)/a(t)$ and the deceleration parameter $q$ defined by,  
\begin{equation}
q=-\frac{a(t)\ddot{a}(t)}{\dot{a}(t)^{2}}.
\end{equation}
Definition \ref{regular}c may be expressed as $q\geq-1$ which is equivalent to the condition that $H(t)$ is a non increasing function of $t$. In terms of the dimensionless density parameters, $\Omega_{M}, \Omega_{R}, \Omega_{\Lambda}$ for mass, radiation (and relativisitic matter), and cosmological constant, respectively, $q$ may written as,
\begin{equation}\label{q}
q=\frac{1}{2}(\Omega_{M}-2\Omega_{\Lambda}+2\Omega_{R}).
\end{equation}
Since each of the densities takes values between $0$ and $1$, it follows from Eq.\eqref{q} that $q\geq-1$. We require in addition that $|q|$ be bounded for all $t$, but with no restrictions on the size of the bound.  The present value, $q_{0}$, has been measured as $-0.58$ by the Supernova Cosmology Project \cite{weinberg}.  With this restriction for the spatially flat case, $k=0$, it follows from Theorem \ref{justcontin} that the metric tensor is continuous across $\mathcal{M}^{0}$. \\

\noindent The remaining restriction we place on the scale factor for Theorems \ref{C1}, \ref{C0}, and \ref{C02} may be understood in similar terms.  By Remark \ref{qadot}, a sufficient additional condition for the inequality Eq. \eqref{abound'} is that $\dot{a}(t)$, regarded as a scale factor itself, is strongly regular, so that its own deceleration parameter is bounded.\\

\noindent In the two dimensional case (or on submanifolds of $\overline{\mathcal{M}}$ with fixed angular coordinates), Theorem \ref{lightlike} shows that the cosmological time zero submanifold $\mathcal{M}^{0}$ is lightlike.  The proof makes use of the relationship between the geometry of Robertson-Walker cosmologies and geometrically defined relative velocities previously developed.\\ 

\noindent Not included in our extended cosmologies $\overline{\mathcal{M}}$ are spacetime points corresponding to $\tau=0$ with $\rho=0$, formally the spacetime point of the comoving observer at the big bang.  Such points must be excluded using the methods we employ here, but it is perhaps possible that by including multiple extended Fermi charts for different comoving observers, a point identified as $\tau=0$ for one observer could be included in an extended chart of another comoving observer. 

\section{Appendix}\label{appendix}

In this appendix we collect the proofs of the theorems and lemmas stated in Section \ref{1+1}.\\

\textbf{Proof of Lemma} \ref{t0cont}:

\begin{proof} Let $\tau_0>0$ be fixed. Choose $M>0$ so that $\tau_0<M$. We have that $t_0(\tau_0,\rho_{\mathcal{M}_{\tau_0}})=0$. Let $\epsilon>0$, with $\epsilon<M$. Set

\begin{equation}
\delta_1=\frac{1}{2}\int_0^\epsilon \frac{a(t)}{\sqrt{a^2(M)-a^2(t)}}dt.
\end{equation}
Since $\rho_{\mathcal{M}_{\tau}}$ is differentiable by Eq.\eqref{dradius}, it is continuous in $\tau$, so choose a $\delta_2$ so that $|\tau-\tau_0|<\delta_2$ implies that

\begin{equation}
\left|\int_0^\tau \frac{a(t)}{\sqrt{a^2(\tau)-a^2(t)}}dt-\int_0^{\tau_0}\frac{a(t)}{\sqrt{a^2(\tau_0)-a^2(t)}}dt\right|<\delta_1
\end{equation}
Let $\delta=\min\{\delta_1,\delta_2,M-\tau_0\}$ and suppose that $|\tau-\tau_0|<\delta\leqslant \delta_2$ and $|\rho-\rho_{\mathcal{M}_{\tau_0}}|<\delta\leqslant \delta_1$. From the triangle inequality,
\begin{align}
&|\rho-\rho_{\mathcal{M}_{\tau_0}}|=\left|\int_{t_0}^\tau \frac{a(t)}{\sqrt{a^2(\tau)-a^2(t)}}dt-\int_0^{\tau_0}  \frac{a(t)}{\sqrt{a^2(\tau_0)-a^2(t)}}dt\right|\nonumber\\
&=\left|\int_{0}^\tau \frac{a(t)}{\sqrt{a^2(\tau)-a^2(t)}}dt-\int_0^{\tau_0}  \frac{a(t)}{\sqrt{a^2(\tau_0)-a^2(t)}}dt-\int_0^{t_0}\frac{a(t)}{\sqrt{a^2(\tau)-a^2(t)}}dt\right|\nonumber\\
&\geqslant \left|\int_0^{t_0}\frac{a(t)}{\sqrt{a^2(\tau)-a^2(t)}}dt\right|-\left|\int_{0}^\tau \frac{a(t)}{\sqrt{a^2(\tau)-a^2(t)}}dt-\int_0^{\tau_0}  \frac{a(t)}{\sqrt{a^2(\tau_0)-a^2(t)}}dt\right|.
\end{align}
Hence,
\begin{equation}\label{tintbound}
\begin{split}
&\left|\int_0^{t_0} \frac{a(t)}{\sqrt{a^2(\tau)-a^2(t)}}dt\right|\leq \left|\rho-\rho_{\mathcal{M}_{\tau_0}}\right|+\\ 
&\left|\int_{0}^\tau \frac{a(t)}{\sqrt{a^2(\tau)-a^2(t)}}dt-\int_0^{\tau_0}  \frac{a(t)}{\sqrt{a^2(\tau_0)-a^2(t)}}dt\right|<\delta_1+\delta_1= 2\delta_1.
\end{split}
\end{equation}
Now, since $a$ is an even function and $|\tau-\tau_0|<\delta\leqslant M-\tau_0$ implies that $\tau<M$, we have that
\begin{align}\label{tintbound2}
\int_0^{|t_0|}\frac{a(t)}{\sqrt{a^2(M)-a^2(t)}}dt&\leqslant \int_0^{|t_0|}\frac{a(t)}{\sqrt{a^2(\tau)-a^2(t)}}dt=\left|\int_0^{t_0}\frac{a(t)}{\sqrt{a^2(\tau)-a^2(t)}}dt\right|\nonumber\\
&<2\delta_1=\int_0^\epsilon \frac{a(t)}{\sqrt{a^2(M)-a^2(t)}}dt,
\end{align}
where we used Eq. \eqref{tintbound} in the second line. Since 
\begin{equation}
\int_0^{|t_0|}\frac{a(t)}{\sqrt{a^2(M)-a^2(t)}}dt
\end{equation}
defines a strictly increasing function of $|t_0|$, Eq. \eqref{tintbound2} implies that $|t_0|<\epsilon$.
\end{proof}

\textbf{Proof of Theorem} \ref{gcontin}

\begin{proof}  From Eq. \eqref{newg3}, it suffices to prove continuity of the function $\mathbf{f}(\tau,\rho)\equiv f(\tau, t_{0}(\tau,\rho))$ at any point of the form $(\tau_{0},\rho_{\mathcal{M}_{\tau_{0}}})$ since continuity at all other points in $D$ follows from smoothness of Fermi coordinates on their original charts.  \\

\noindent We first simplify notation by abbreviating the integrand for the function $f(\tau,t_0)$ in Eq.\eqref{gint} for the case $t_{0}\geq0$ as follows,

\begin{equation}
f(\tau,t_{0})=\int_{t_{0}}^{\tau}h(\tau,t_{0},t)dt.
\end{equation}
Now define a new function of two independent variables $\tau$ and $x$, by,
\begin{equation}\label{defh}
f(\tau,x)=\int_{x}^{\tau}h(\tau,x,t)dt=\int_{x}^{\tau}\frac{\ddot{a}(t)}{\dot{a}(t)^{2}}\left(\frac{\sqrt{a^{2}(\tau)-a^{2}(x)}}{\sqrt{a^{2}(\tau)-a^{2}(t)}}-1\right)dt
\end{equation}
with the restriction $\tau>x\geq0$.  Since $t_{0}(\tau_{0},\rho_{\mathcal{M}_{\tau_{0}}})=0$, our plan of proof is first to show that $f(\tau,x)$ is continuous at any point of the form $(\tau,0)$. Then using Lemma \ref{t0cont},  the composition $f(\tau, t_{0}(\tau,\rho))$ must be continuous  at any point of the form $(\tau_{0},\rho_{\mathcal{M}_{\tau_{0}}})$ in the domain where $t_{0}\geq0$.  It then follows directly from Eq. \eqref{gint} and from the easily verified result that $f(\tau,0)$ is a continuous function of $\tau$ that the restriction $t_{0}\geq0$ may be removed, and the theorem will be established. \\

\noindent If $\tau\geq\tau_{0}>0$, we may assume with no loss of generality that $x<\tau_{0}$. Then from the triangle inequality,
\begin{multline}\label{gttcont}
|f(\tau,x)-f(\tau_{0},0)|\leq\int_{x}^{\tau_{0}}|h(\tau,x,t)-h(\tau_{0},0,t)|dt \\+\left|\int_{\tau_{0}}^{\tau}h(\tau,x,t)dt\right|
+\left|\int_{0}^{x}h(\tau_{0},0,t)dt\right|.
\end{multline}
The third term can be made arbitrarily small by choosing $x$ sufficiently close to zero because $h(\tau_{0},0,t)$ is integrable. The second term on the right is small for $\tau$ sufficiently close to $\tau_{0}$, uniformly in $x$ according to Lemma \ref{continuitylemma}(a) with $\ell(t)=\ddot{a}(t)/\dot{a}(t)^{2}$. The first term on the right side of Eq.\eqref{gttcont} can be made arbitrarily small by choosing $\tau$ sufficiently close to $\tau_{0}$ and $x$ sufficiently close to $0$ by using the fact that $h$ is continuous jointly in all its variables and that $h(\tau_{0},0,t)$ and $h(\tau,x,t)$ are both integrable on $(0,\tau_{0})$.   \\

\noindent  For the other case,  $\tau<\tau_{0}$, the right side of Eq. \eqref{gttcont} must be modified by replacing $\tau_{0}$ by $\tau$ in the upper limit of integration of the first term, interchanging  $\tau_{0}$ and $\tau$ and setting $x=0$ in the second term, and leaving the third term unchanged.  All three terms may then be bounded as before. Thus, $f(\tau,\rho)$ is continuous at any point of the form $(\tau_{0},\rho_{\mathcal{M}_{\tau_{0}}})$ and therefore $g_{\tau\tau}$ is continuous on $D$.
\end{proof}

\textbf{Proof of Lemma }\ref{c1lemma1}

\begin{proof} 

\noindent Let $\tau>0$. First we show that Eq. \eqref{dt0g} holds for $t_0\in(0,\tau)$. Choose $\delta>0$ so that $[t_0-\delta,t_0+\delta]\subset(0,\tau)$. As in the proof of Theorem \ref{gcontin}, let $h(\tau,t_0,t)$ denote the integrand of Eq. \eqref{gint}. For any $\Delta t_{0}$ with $0<|\Delta t_{0}|<\delta$, we have that
\begin{multline}
\frac{f(\tau,t_0+\Delta t_{0})-f(\tau,t_0)}{\Delta t_{0}}\\ =\int_{t_0}^\tau \frac{h(\tau,t_0+\Delta t_{0},t)-h(\tau,t_0,t)}{\Delta t_{0}}dt-\int_{t_0}^{t_0+\Delta t_{0}}\frac{h(\tau,t_0+\Delta t_{0},t)}{\Delta t_{0}}dt
\end{multline}
By the mean value theorem there exists $\zeta_{\Delta t_{0}}$ and $\xi_{\Delta t_{0}}$ between $t_0$ and $t_0+\Delta t_{0}$ such that
\begin{equation}
\int_{t_0}^{t_0+\Delta t_{0}}\frac{h(\tau,t_0+\Delta t_{0},t)}{\Delta t_{0}}dt=h(\tau,t_0+\Delta t_{0},\zeta_{\Delta t_{0}}),
\end{equation}
and
\begin{multline}
\frac{h(\tau,t_0+\Delta t_{0},t)-h(\tau,t_0,t)}{\Delta t_{0}}=\partial_{t_0}h(\tau,\xi_{\Delta t_{0}},t)\\=-\frac{a(\xi_{\Delta t_{0}})\dot{a}(\xi_{\Delta t_{0}})}{\sqrt{a^{2}(\tau)-a^{2}(\xi_{\Delta t_{0}})}}\frac{\ddot{a}(t)}{\dot{a}(t)^{2}}\frac{1}{\sqrt{a^{2}(\tau)-a^{2}(t)}}
\end{multline}
So, we can write
\begin{multline}\label{L2int}
\frac{f(\tau,t_0+\Delta t_{0})-f(\tau,t_0)}{\Delta t_{0}}=-\frac{a(\xi_{\Delta t_{0}})\dot{a}(\xi_{\Delta t_{0}})}{\sqrt{a^{2}(\tau)-a^{2}(\xi_{\Delta t_{0}})}}\int_{t_0}^\tau\frac{\ddot{a}(t)}{\dot{a}(t)^{2}}\frac{dt}{\sqrt{a^{2}(\tau)-a^{2}(t)}}\\-h(\tau,t_0+\Delta t_{0},\zeta_{\Delta t_{0}})
\end{multline}
Now let, $\Delta t_{0}\to0$, so that $\zeta_{\Delta t_{0}}\to t_0$ and $\xi_{\Delta t_{0}}\to t_0$. Eq. \eqref{dt0g} follows from Eq. \eqref{L2int} by continuity. Now if $t_0\in(-\tau,0)$, we have that
\begin{equation}
\partial_{t_0}f(\tau,t_0)=\partial_{t_0}[2f(\tau,0)-f(\tau,-t_0)]=\partial_{t_0}f(\tau,-t_0),
\end{equation}
which completes the proof since $-t_0=|t_0|$.
\end{proof}

\textbf{Proof of Lemma} \ref{c1lemma2}

\begin{proof}
(a) By assumption $\ddot{a}(t)/\dot{a}^2(t)\geqslant0$ on $(0,\epsilon)$. It follows that for any $t_0\in(0,\epsilon)$,
\begin{align}\label{divrintp}
\int_{t_{0}}^{\tau}\frac{\ddot{a}(t)}{\dot{a}(t)^{2}}\frac{dt}{\sqrt{a^{2}(\tau)-a^{2}(t)}}&\geqslant\frac{1}{a(\tau)}\int_{t_0}^\epsilon\frac{\ddot{a}(t)}{\dot{a}(t)^{2}}dt+\int_\epsilon^\tau\frac{\ddot{a}(t)}{\dot{a}(t)^{2}}\frac{dt}{\sqrt{a^{2}(\tau)-a^{2}(t)}}\nonumber\\
&= \frac{1}{a(\tau)}\left(\frac{1}{\dot{a}(t_0)}-\frac{1}{\dot{a}(\epsilon)}\right)+\int_\epsilon^\tau\frac{\ddot{a}(t)}{\dot{a}(t)^{2}}\frac{dt}{\sqrt{a^{2}(\tau)-a^{2}(t)}}.
\end{align}
Since the integral on the second line of Eq. \eqref{divrintp} is finite, the second line of Eq. \eqref{divrintp} is seen to diverge to $\infty$ as $t_0\to 0^+$. Hence Eq. \eqref{divrint} follows.\newline
(b) Observe first, using regularity, that for $0<\delta<\tau$,

\begin{equation}
\int_{\delta}^{\tau}\frac{\ddot{a}(t)}{\dot{a}^2(t)}\frac{dt}{\sqrt{a^{2}(\tau)-a^{2}(t)}}\leq\frac{1}{a^{2}(\delta)}\int_{\delta}^{\tau}\frac{a(t)}{\sqrt{a^{2}(\tau)-a^{2}(t)}}\,dt<\infty,
\end{equation} 
by Eq.\eqref{radius}.  Now from continuity and since $\dot{a}(0^+)>0$, the integrand on the left side is bounded near zero, and the result follows. 
\end{proof}

\begin{remark}\label{derivative}
The following observation will be useful in what follows. Suppose that $f:[a,b]\to\mathbb{R}$ is continuous, and continuously differentiable on $(a,b)$ except perhaps at some point $x_0\in(a,b)$. If $\lim_{x\to x_0}f'(x)=L$ exists, then $f'(x_0)=\lim_{x\to x_0}\frac{f(x)-f(x_0)}{x-x_0}=L$ by L'H\^{o}pital's rule. Hence $f$ is in fact continuously differentiable on $(a,b)$.
\end{remark}

\textbf{Proof of Theorem} \ref{drhometric'}

\begin{proof}
From Eq.\eqref{newg3} it suffices to prove that $f(\tau,t_0(\tau,\rho))=\mathbf{f}(\tau,\rho)$ is differentiable with respect to $\rho$ in $D$, where $f(\cdot,\cdot)$ is given by Eq. \eqref{gint}. Eq. \eqref{drhometric} then follows from Eq. \eqref{newg3} and the chain rule.  First note that from Eq. \eqref{properAlt2}, for $t_0\neq0$,

\begin{equation}\label{dtdrho}
\pder{t_{0}}{\rho}=-\frac{\sqrt{a^{2}(\tau)-a^{2}(t_{0})}}{a(t_{0})}.
\end{equation}

\noindent For $(\tau,\rho)\neq (\tau,\rho_{\mathcal{M}_\tau})$ in $D$, by Lemma \ref{c1lemma1}, Eq. \eqref{dtdrho} and the chain rule we have that,

\begin{equation}
\partial_{\rho}\mathbf{f}(\tau,\rho)=\partial_{t_{0}}f(\tau,t_{0}(\tau,\rho))\cdot\partial_{\rho}t_{0}(\tau,\rho).
\end{equation}

\noindent It remains to prove that for any $\tau>0$ that $\mathbf{f}(\tau,\rho)$ is differentiable with respect to $\rho$ at $(\tau,\rho_{\mathcal{M}_\tau})$ (i.e. at $t_0=0$). To this end, we will prove the existence of the limit
\begin{equation}
\lim_{\rho\to\rho_{\mathcal{M}_\tau}}\partial_\rho \mathbf{f}(\tau,\rho)=\lim_{t_0\to0}\partial_{t_{0}}f(\tau,t_{0}(\tau,\rho))\cdot\partial_{\rho}t_{0}(\tau,\rho).
\end{equation}
\noindent The desired result will then follow by Remark \ref{derivative}.  By Lemma \ref{c1lemma1},
\begin{align}
\partial_{t_{0}}f(\tau,t_{0})\cdot\partial_{\rho}t_{0}&=-\frac{a(t_0)\dot{a}(|t_{0}|)}{\sqrt{a^2(\tau)-a^2(t_0)}}\int_{|t_{0}|}^{\tau}\frac{\ddot{a}(t)}{\dot{a}(t)^{2}}\frac{dt}{\sqrt{a^{2}(\tau)-a^{2}(t)}}\times\nonumber\\
&-\frac{\sqrt{a^2(\tau)-a^2(t_0)}}{a(t_0)}\nonumber\\
&=\dot{a}(|t_0|)\int_{|t_{0}|}^{\tau}\frac{\ddot{a}(t)}{\dot{a}(t)^{2}}\frac{dt}{\sqrt{a^{2}(\tau)-a^{2}(t)}}.
\end{align}
We now consider the two sets of conditions from Lemma \ref{c1lemma2}.\newline
\noindent (a) If $\dot{a}(0^+)=0$ then by Lemma \ref{c1lemma2} and L'H\^{o}pital's rule,
\begin{multline}\label{limrhof}
\lim_{\rho\to\rho_{\mathcal{M}_\tau}}\partial_{\rho}\mathbf{f}(\tau,\rho)=\lim_{t_0\to0}\dot{a}(|t_{0}|)\int_{|t_{0}|}^{\tau}\frac{\ddot{a}(t)}{\dot{a}(t)^{2}}\frac{1}{\sqrt{a^{2}(\tau)-a^{2}(t)}}dt\\=\lim_{t_0\to0}\frac{\int_{|t_{0}|}^{\tau}\frac{\ddot{a}(t)}{\dot{a}(t)^{2}}\frac{1}{\sqrt{a^{2}(\tau)-a^{2}(t)}}dt}{1/\dot{a}(|t_0|)}=\frac{1}{a(\tau)}
\end{multline}
(b) If $\dot{a}(0^+)>0$ then by Lemma \ref{c1lemma2}, 
\begin{multline}
\lim_{\rho\to\rho_{\mathcal{M}_\tau}}\partial_{\rho}\mathbf{f}(\tau,\rho)=\lim_{t_0\to0}\dot{a}(|t_{0}|)\int_{|t_{0}|}^{\tau}\frac{\ddot{a}(t)}{\dot{a}(t)^{2}}\frac{1}{\sqrt{a^{2}(\tau)-a^{2}(t)}}dt\\
=\dot{a}(0^+)\int_0^\tau\frac{\ddot{a}(t)}{\dot{a}^2(t)}\frac{1}{\sqrt{a^2(\tau)-a^2(t)}}dt
\end{multline}

\end{proof}

\textbf{Proof of Theorem} \ref{dtaug}

\begin{proof} From Eq. \eqref{newg3}, it suffices to prove differentiability of the function $\mathbf{f}(\tau,\rho)\equiv f(\tau, t_{0}(\tau,\rho))$ at any point of the form $(\tau_{0},\rho_{\mathcal{M}_{\tau_{0}}})$ since differentiability at all other points in $D$ follows from smoothness of Fermi coordinates in their original charts. By Remark \ref{derivative}, differentiability at $(\tau_{0},\rho_{\mathcal{M}_{\tau_{0}}})$ will follow from the existence of the limit,
\begin{multline}\label{tauderivlimit2}
\lim_{\tau\rightarrow\tau_{0}}\der{f(\tau, t_{0}(\tau,\rho_{\mathcal{M}_{\tau_{0}}}))}{\tau}= \lim_{\tau\rightarrow\tau_{0}}\pder{f(\tau, t_{0})}{\tau}
+ \lim_{\tau\rightarrow\tau_{0}}\pder{f(\tau, t_{0})}{t_{0}}\pder{t_{0}(\tau,\rho_{\mathcal{M}_{\tau_{0}}}))}{\tau}.
\end{multline}
This is established by Lemmas \ref{taulimlemma1} and \ref{taulimlemma2} below which establish the existence of the limits on the right side of Eq. \eqref{tauderivlimit2} for an arbitrary $\tau_0>0$.\\

\end{proof}

\textbf{Proof of Lemma} \ref{firstlemmatau}

\begin{proof}
\noindent To compute $\partial f/\partial \tau$ in the case that $t_0>0$, we make the change of variable $\sigma=(a(\tau)/a(t))^2$ and find that,

\begin{equation}\label{fsub}
f(\tau,t_0)=-\frac{a(t_0)}{2}\int_1^{\sigma(\tau)}\ddot{b}\left(\frac{a(\tau)}{\sqrt{\sigma}}\right)\left(\frac{\sqrt{\sigma(\tau)-1}}{\sqrt{\sigma-1}}-\frac{\sqrt{\sigma(\tau)}}{\sqrt{\sigma}}\right)\frac{d\sigma}{\sigma},
\end{equation}
where $\sigma(\tau)=(a(\tau)/a(t_0))^2$ and $b(t)$ is the inverse function of $a(t)$.  Applying the Dominated Convergence theorem, we calculate,

\begin{equation}\label{dfdtpositive}
\pder{f}{\tau}(\tau,t_0)=I_1(\tau,t_0)+I_2(\tau,t_0).
\end{equation}

\noindent If $t_0<0$, then $f(\tau,t_0)=2f(\tau,0)-f(\tau,-t_0)$. Changing the integration variable to $\sigma$ again, we find that
\begin{equation}
f(\tau,0)=-\frac{a(\tau)}{2}\int_0^\infty \ddot{b}\left(\frac{a(\tau)}{\sqrt{\sigma}}\right)\left[\sqrt{\frac{\sigma}{\sigma-1}}-1\right]\frac{d\sigma}{\sigma^{3/2}}.
\end{equation}

\noindent Applying the dominated convergence theorem we verify Eq. \eqref{partialtauf0}. Then by Eq. \eqref{dfdtpositive}, we have

\begin{equation}
\pder{f}{\tau}(\tau,t_0)=\pder{}{\tau}(2f(\tau,0)-f(\tau,-t_0))=2\pder{f}{\tau}(\tau,0)-I_1(\tau,-t_0)-I_2(\tau,-t_0).
\end{equation}

\end{proof}

\textbf{Proof of Lemma} \ref{taulimlemma1}

\begin{proof}
First consider $\tau>\tau_0$. It is easy to show that in this case $t_0(\tau)>0$. Using the hypotheses and following the notation of Lemma \ref{firstlemmatau} we may bound the integrand of $I_1(\tau,t_0(\tau))$ as follows,
\begin{align}
I_{[t_0(\tau),\tau]}&\left|3\frac{\ddot{a}^2(t)a(t)}{\dot{a}^4(t)}-\frac{\dddot{{}a}(t)a(t)}{\dot{a}^3(t)}\right|\left[\frac{\sqrt{a^2(\tau)-a^2(t_0(\tau))}}{\sqrt{a^2(\tau)-a^2(t)}}-1\right]&\nonumber\\
&\leqslant (3K^2+C)I_{[0,\tau]}\frac{1}{a(t)}\left[\frac{a(\tau)}{\sqrt{a^2(\tau)-a^2(t)}}-1\right]\nonumber\\
&\leqslant\frac{3K^2+C}{\dot{a}(\tau)}I_{[0,\tau]}\frac{\dot{a}(t)/a(\tau)}{\sqrt{1-\frac{a^2(t)}{a^2(\tau)}}\left(1+\sqrt{1-\frac{a^2(t)}{a^2(\tau)}}\right)},
\end{align}
\noindent where $I_{[a,b]}$ is the indicator function for $[a,b]$, and where in the last line we have used the fact that the Hubble parameter, $H(t)$, is a decreasing function of $t$, as in Eq. \eqref{long}. A direct calculation shows that
\begin{multline}\label{i1boundlim}
\lim_{\tau\to\tau_0}\frac{1}{\dot{a}(\tau)}\int_0^\tau\frac{\dot{a}(t)/a(\tau)}{\sqrt{1-\frac{a^2(t)}{a^2(\tau)}}\left(1+\sqrt{1-\frac{a^2(t)}{a^2(\tau)}}\right)}dt\\=\frac{1}{\dot{a}(\tau_0)}\int_0^{\tau_0}\frac{\dot{a}(t)/a(\tau_0)}{\sqrt{1-\frac{a^2(t)}{a^2(\tau_0)}}\left(1+\sqrt{1-\frac{a^2(t)}{a^2(\tau_0)}}\right)}dt.
\end{multline}
So by Eq.\eqref{i1boundlim} and the generalized dominated convergence theorem \cite{royden},
\begin{equation}\label{i1lim}
\lim_{\tau\to\tau_0^+}I_1(\tau,t_0(\tau))=I_1(\tau_0,0).
\end{equation}
Let $\delta\in(0,\tau_0)$ and $M>\tau_0$. We can assume that $\tau$ is sufficiently close to $\tau_0$ so that $t_0(\tau)<\delta$ and $\tau<M$. Using the hypotheses of the lemma we may bound the integrand of $I_2(\tau,t_0(\tau))$ as follows,
\begin{align}\label{i2boundlim}
&I_{[t_0(\tau),\tau]}\left|\frac{\ddot{a}(t)}{\dot{a}^2(t)}\right|\left[\frac{a^2(\tau)}{\sqrt{a^2(\tau)-a^2(t)}\sqrt{a^2(\tau)-a^2(t_0)}}-1\right]\nonumber\\
&\leqslant I_{[0,\delta]}\frac{K}{a(t)}\left[\frac{a^2(\tau)}{a^2(\tau)-a^2(t)}-1\right]+I_{[\delta,\tau]}\left|\frac{\ddot{a}(t)}{\dot{a}^2(t)}\right|\left[\frac{a^2(\tau)}{\sqrt{a^2(\tau)-a^2(t)}\sqrt{a^2(\tau)-a^2(t_0)}}-1\right]\nonumber\\
&\leqslant I_{[0,\delta]}K\frac{a(\delta)}{a^2(\tau)-a^2(\delta)}+I_{[\delta,\tau]}D\frac{a^2(\tau)}{\sqrt{a^2(\tau)-a^2(\delta)}}\frac{\dot{a}(t)}{\sqrt{a^2(\tau)-a^2(t)}},
\end{align}
where $|\ddot{a}(t)/\dot{a}^3(t)|\leqslant D$ on $[\delta,M]$. Then by Eq. \eqref{i2boundlim} and the generalized dominated convergence theorem, we have that
\begin{equation}\label{i2lim}
\lim_{\tau\to\tau_0^+}I_2(\tau,t_0(\tau))=I_2(\tau_0,0).
\end{equation}
So by Eqs. \eqref{i1lim} and \eqref{i2lim} we have that
\begin{equation}
\lim_{\tau\to\tau_0^+}\pder{f}{\tau}(\tau,t_0(\tau))=I_1(\tau_0,0)+I_2(\tau_0,0)=\pder{f}{\tau}(\tau_0,0).
\end{equation}
If $\tau<\tau_0$ then $t_0(\tau)<0$ and similar arguments used for the $\tau>\tau_0$ case above show that 
\begin{equation}\label{i12limnegative}
\lim_{\tau\to\tau_0^-} I_1(\tau,-t_0(\tau))+I_2(\tau,-t_0(\tau))=I_1(\tau_0,0)+I_2(\tau_0,0).
\end{equation}
Using the hypotheses of the theorem we may bound the integrand of Eq. \eqref{partialtauf0} in a similar way as we did for $I_1$ and $I_2$ and apply the generalized dominated convergence theorem to calculate
\begin{equation}\label{limdftau0}
\lim_{\tau\to\tau_0}\pder{f}{\tau}(\tau,0)=\pder{f}{\tau}(\tau_0,0).
\end{equation}
By Eqs. \eqref{i12limnegative} and \eqref{limdftau0},
\begin{multline}
\lim_{\tau\to\tau_0^-}\pder{f}{\tau}(\tau,t_0(\tau))\\=\lim_{\tau\to\tau_0^-}\left(2\pder{f}{\tau}(\tau,0)-I_1(\tau,-t_0(\tau))-I_2(\tau,t_0(\tau))\right)=\pder{f}{\tau}(\tau_0,0).
\end{multline}
\end{proof}

\textbf{Proof of Lemma }\ref{taulimlemma2}

\begin{proof}
\noindent A convenient expression for $\pder{t_0}{\tau}$, valid for $t_0\neq0$ is obtained from the relationship,

\begin{align}
\rho=\int_{t_0}^{\tau}\frac{a(t)}{\sqrt{a^2(\tau)-a^2(t)}}dt&=\int_0^\tau\frac{a(t)}{\sqrt{a^2(\tau)-a^2(t)}}dt-\int_0^{t_0}\frac{a(t)}{\sqrt{a^2(\tau)-a^2(t)}}dt\nonumber\\
&\equiv\rho_{\mathcal{M}_\tau}-G(\tau,t_0).
\end{align}
So by the chain rule,

\begin{equation}\label{taudp1}
0=\der{\rho_{\mathcal{M}_\tau}}{\tau}(\tau)-\pder{G}{\tau}(\tau,t_0)-\pder{G(\tau,t_0)}{t_0}(\tau,t_0)\pder{t_0}{\tau}
\end{equation}
or,
\begin{equation}\label{taudp}
\pder{t_0}{\tau}=\frac{\sqrt{a^2(\tau)-a^2(t_0)}}{a(t_0)}\left[\der{\rho_{\mathcal{M}_\tau}}{\tau}(\tau)-\pder{G}{\tau}(\tau,t_0)\right].
\end{equation}
where $d\rho_{\mathcal{M}_\tau}/d\tau$ is given by Eq.\eqref{dradius} and a calculation shows that,
\begin{equation}\label{dgdtau}
\pder{G}{\tau}(\tau,t_0)=-a(\tau)\dot{a}(\tau)\int_0^{t_0}\frac{a(t)}{(a^2(\tau)-a^2(t))^{3/2}}dt.
\end{equation}
It follows easily from the generalized dominated convergence theorem and the continuity of $\rho_{\mathcal{M}_\tau}$ that
\begin{equation}\label{Grholimit}
\lim_{\tau\to\tau_0}\left[\der{\rho_{\mathcal{M}_\tau}}{\tau}(\tau)-\pder{G}{\tau}(\tau,t_0(\tau))\right]=\der{\rho_{\mathcal{M}_\tau}}{\tau}(\tau_0).
\end{equation}
By Lemma \ref{c1lemma1} and Eq. \eqref{taudp},
\begin{multline}
\pder{f}{t_{0}}(\tau,t_0(\tau))\pder{t_0}{\tau}(\tau,t_0(\tau))\\=\frac{\dot{a}(|t_0(\tau)|)}{\sqrt{a^2(\tau)-a^2(t_0(\tau))}}\int_{|t_0(\tau)|}^\tau\frac{\ddot{a}(t)}{\dot{a}^2(t)}\frac{1}{\sqrt{a^2(\tau)-a^2(t)}}dt\left[\der{\rho_{\mathcal{M}_\tau}}{\tau}(\tau)-\pder{G}{\tau}(\tau,t_0(\tau))\right].
\end{multline}
We now consider the two conditions from Lemma \ref{c1lemma2}.\newline
(a) If $a(0^+)=0$, then a calculation using the generalized dominated convergence theorem \cite{royden}, L'H\^{o}pital's rule and Eq. \eqref{Grholimit} shows that
\begin{equation}
\lim_{\tau\to\tau_0}\pder{f}{\tau}(\tau,t_0(\tau))\pder{t_0}{\tau}(\tau,t_0(\tau))=\frac{1}{a(\tau_0)}\der{\rho_{\mathcal{M}_\tau}}{\tau}(\tau_0).
\end{equation}
(b) If $a(0^+)>0$, then a calculation using the generalized dominated convergence theorem shows that
\begin{multline}
\lim_{\tau\to\tau_0}\pder{f}{\tau}(\tau,t_0(\tau))\pder{t_0}{\tau}(\tau,t_0(\tau))\\=\dot{a}(0^+)\der{\rho_{\mathcal{M}_\tau}}{\tau}(\tau_0)\int_0^\tau\frac{\ddot{a}(t)}{\dot{a}^2(t)}\frac{1}{\sqrt{a^2(\tau)-a^2(t)}}dt,
\end{multline}
and this expression is finite.
\end{proof}

\textbf{Proof of Theorem} \ref{rhoderivg}

\begin{proof}
By Eq. \eqref{drhometric}, continuity of $\dot{a}(t)$, and the continuity of $\sqrt{-g_{\tau\tau}}$ established by  Theorem \ref{gcontin}, it suffices to show that $\partial_\rho \mathbf{f}(\tau,\rho)$ is continuous at points of the form $(\tau_0,\rho_{\mathcal{M}_{\tau_0}})$ in $D$. For the duration of this proof, let,

\begin{equation}\label{F0}
F_{0}(\tau,x)=\dot{a}(x)\int_{x}^\tau \frac{\ddot{a}(t)}{\dot{a}^2(t)}\frac{dt}{\sqrt{a^2(\tau)-a^2(t)}}dt\equiv \dot{a}(x)\int_{x}^\tau h(\tau,t)dt.
\end{equation}
From Eq.\eqref{drhometric2}, $\partial_\rho \mathbf{f}(\tau,\rho)=F_{0}(\tau,|t_{0}(\tau,\rho)|)$ for $\rho\neq \rho_{\mathcal{M}_\tau}$.
When $\rho\ = \rho_{\mathcal{M}_\tau}$ there are two cases to consider:
(a) $\dot{a}(0^+)>0$ and (b) $\dot{a}(0)=0$.  For each of these cases we will define $F_{0}(\tau,0)=\partial_\rho \mathbf{f}(\tau,\rho_{\mathcal{M}_\tau})$ so that $F_{0}(\tau,x)$ is defined on the set $\{(\tau,x): \tau\geq x\geq0\}$.  Since $t_{0}(\tau,\rho)=0$ if and only if $\rho = \rho_{\mathcal{M}_\tau}$ and $t_{0}(\tau,\rho)$ is continuous on $D$ by Lemma \ref{t0cont}, it suffices to show that $F_{0}(\tau,x)$ is continuous at $(\tau_{0},0)$ for any $\tau_{0}>0$. \\

\noindent For case (a), denote $\dot{a}(0^+)$ by $\dot{a}(0)$ for ease of notation. Then $\partial_\rho \mathbf{f}(\tau,\rho_{\mathcal{M}_\tau})=F_{0}(\tau,0)$. Since $\dot{a}(|t_0|)$ is continuous everywhere, we need only show that $\int_{x}^\tau h(\tau,t)dt$ is continuous at $(\tau,x)=(\tau_{0},0)$. \\ 

\noindent Let $\tau_0>0$ be fixed and first suppose that $\tau>\tau_0$. From the triangle inequality and using Lemma \ref{c1lemma2}(b),
\begin{multline}\label{rhocontineq2}
\left|\int_{x}^\tau h(\tau,t)dt-\int_0^{\tau_0}h(\tau_0,t)dt\right|\leqslant \left|\int_{\tau_0}^\tau h(\tau,t)dt\right|\\+\int_0^{\tau_0}|h(\tau,t)-h(\tau_0,t)|dt+\left|\int_0^xh(\tau_0,t)dt\right|.
\end{multline}
The first term on the right can be made arbitrarily small for $\tau$ sufficiently close to $\tau_0$ by Lemma \ref{continuitylemma}(b). The last term in this inequality can made arbitrarily small for sufficiently small $x$ by the integrability of $h(\tau_0,t)$. The middle term can be made small for all $\tau$ sufficiently close to $\tau_0$ since $h(\tau,t)$ is uniformly continuous in both variables for $\tau$ in a neighborhood of $\tau_{0}$ and $t$ restricted to any interval of the form $[0,\tau_{0}-\epsilon]$ for $\epsilon>0$, and because $h(\tau_{0},t)$ is integrable on $[0,\tau_{0}]$. The case $\tau<\tau_0$ is similar.\\

\noindent For case (b), let $F_{0}(\tau,0)\equiv\lim_{x\to0^{+}}F_{0}(\tau,x)$.  Then $F_{0}(\tau,0) =1/a(\tau)=\partial_\rho \mathbf{f}(\tau,\rho_{\mathcal{M}_\tau})$ by Theorem \ref{drhometric'}, so,
\begin{equation}\label{caseb}
|F_{0}(\tau,x)-F_{0}(\tau_{0},0)| \leq |F_{0}(\tau,x)-F_{0}(\tau_{0},x)|+\left|F_{0}(\tau_{0},x)-\frac{1}{a(\tau_{0})}\right|.
\end{equation}
The second term on the right side of Eq.\eqref{caseb} can be made arbitrarily small for all $x$ sufficiently close to zero. To show that the first term on the right side can be made arbitrarily small for all $(\tau,x)$ sufficiently close to $(\tau_{0},0)$, we assume first that $\tau>\tau_{0}$ so that,
\begin{equation}\label{caseb2}
|F_{0}(\tau,x)-F_{0}(\tau_{0},x)|\leq \dot{a}(x)\int_{\tau_{0}}^\tau |h(\tau,t)|dt+\dot{a}(x)\int_{x}^{\tau_{0}}|h(\tau,t)-h(\tau_{0},t)|dt.
\end{equation}
The first term on the right side of Eq.\eqref{caseb2} can be made arbitrarily small for $\tau$ sufficiently close to $\tau_0$ by Lemma \ref{continuitylemma}(b) and the boundedness of $\dot{a}(x)$.  For the second term, choose $\epsilon$ with $0<\epsilon<\tau_{0}$ so that in accordance with Lemma \ref{c1lemma2}(a), $\ddot{a}(t)>0$ on $(0,\epsilon)$. For $x<\epsilon$,
\begin{multline}\label{caseb3}
\dot{a}(x)\int_{x}^{\tau_{0}}|h(\tau,t)-h(\tau_{0},t)|dt=\dot{a}(x)\int_{\epsilon}^{\tau_{0}}|h(\tau,t)-h(\tau_{0},t)|dt\\
+\dot{a}(x)\int_{x}^{\epsilon}|h(\tau,t)-h(\tau_{0},t)|dt
\end{multline}
The first integral on the right side of Eq.\eqref{caseb3} is bounded by Lemma \ref{continuitylemma}(b) and because $|h(\tau,t)|\leq|h(\tau_{0},t)|$ for $\tau\geq\tau_{0}$.  Thus the first term on the right converges to zero uniformly in $\tau$ as $x\to0^{+}$.  By continuity, given any $\delta>0$, the integrand in the second term on the right is bounded by $\frac{\ddot{a}(t)}{\dot{a}^2(t)}\delta$ for $\tau$ sufficiently close to $\tau_{0}$, independent of $x$, and then,
\begin{equation}
\dot{a}(x)\int_{x}^{\epsilon}|h(\tau,t)-h(\tau_{0},t)|dt<\delta\dot{a}(x)\int_{x}^{\epsilon}\frac{\ddot{a}(t)}{\dot{a}^2(t)}dt=\delta\left[1-\frac{\dot{a}(x)}{\dot{a}(\epsilon)}\right]<\delta.
\end{equation}
because $\dot{a}(t)$ is increasing on $(0,\epsilon)$.  The argument for the case that $\tau<\tau_{0}$ is similar.  Thus, $F_{0}(\tau,x)$ is continuous for $\tau\geq x\geq 0$ and it follows that $\partial_\rho \mathbf{f}(\tau,\rho)$ is continuous on $D$ and therefore $\partial_\rho g_{\tau\tau}$ is continuous on $D$.
\end{proof}

\textbf{Proof of Theorem} \ref{tauderivg}

\begin{proof} By Eq. \eqref{tauderivformula}, the smoothness of $a(\tau)$, and the continuity of $g_{\tau\tau}$ established by  Theorem \ref{gcontin}, we need only show that $\partial_{\tau} \mathbf{f}(\tau,\rho)$ is continuous at points of the form $(\tau_0,\rho_{\mathcal{M}_{\tau_0}})$ in $D$.\\

\noindent From Eq.\eqref{dtauf} it suffices to show that $\partial_{\tau}f(\tau,t_0)$ and $\partial_{t_0}f(\tau,t_0)\partial_\tau t_0(\tau,t_0)$ are continuous at points of the form $(\tau_0,\rho_{\mathcal{M}_{\tau_0}})$, i.e., exactly at the points where the function $t_{0}$ vanishes.\\

\noindent Since $t_{0}=t_{0}(\tau, \rho)$ is continuous on $D$ by Lemma \ref{t0cont}, it is sufficient to show that $\partial_{\tau}f(\tau,t_0)$ and $\partial_{t_0}f(\tau,t_0)\partial_\tau t_0(\tau,t_0)$ are continuous as functions of $\tau$ and $t_{0}$ at points of the form $(\tau,t_{0})=(\tau,0)$.\\

\noindent We first consider $\partial_{\tau}f(\tau,t_0)$.  From Eq.\eqref{limdftau0}, $\partial_{\tau}f(\tau,0)$ is a continuous function of $\tau$, so in light of Eq.\eqref{ftauderivative}, it is enough to prove that the functions $I_{1}(\tau,x)$ and $I_{2}(\tau,x)$ given by Eqs.\eqref{I1}, and \eqref{I2} and defined on the set $\{(\tau,x): \tau\geq x\geq0\}$, are continuous at any point of the form $(\tau_{0},0)$.\\ 

\noindent To show that $I_{1}(\tau,x)$ is continuous at $(\tau_{0},0)$ we show that the same is true of the integral expression, $F_{1}(\tau,x)$ given by,
\begin{equation}
\begin{split}\label{I1a}
F_{1}(\tau,x)&\equiv-\frac{a(\tau)}{\dot{a}(\tau)}I_1(\tau,x)\\
&=\int_{x}^\tau\left[3\frac{\ddot{a}^2(t)a(t)}{\dot{a}^4(t)}-\frac{\dddot{{}a}(t)a(t)}{\dot{a}^3(t)}\right]\left(\frac{\sqrt{a^2(\tau)-a^2(x)}}{\sqrt{a^2(\tau)-a^2(t)}}-1\right)dt\\
&\equiv\int_{x}^\tau\ell_{1}(t)\left(\frac{\sqrt{a^2(\tau)-a^2(x)}}{\sqrt{a^2(\tau)-a^2(t)}}-1\right)dt\\
&\equiv\int_{x}^{\tau}h_{1}(\tau,x,t)dt,
\end{split}
\end{equation}
in analogy with Eq.\eqref{defh}. The proof of continuity of $F_{1}(\tau,x)$ at $(\tau_{0},0)$ proceeds exactly as in Eq.\eqref{gttcont} and the two paragraphs following it, with $f$ replaced by $F_{1}$ and $h$ replaced by $h_{1}$. \\

\noindent We proceed in a similar fashion for the function $I_{2}(\tau,x)$.  Define
\begin{equation}
\begin{split}\label{I2a}
F_{2}(\tau,x)&\equiv\frac{a(\tau)}{\dot{a}(\tau)}I_2(\tau,x)\\
&=\int_{x}^\tau\frac{\ddot{a}(t)}{\dot{a}^2(t)}\left[\frac{a^2(\tau)}{\sqrt{a^2(\tau)-a^2(t)}\sqrt{a^2(\tau)-a^2(x)}}-1\right]dt\\
&\equiv\int_{x}^{\tau}h_{2}(\tau,x,t)dt.
\end{split}
\end{equation}
As a preliminary estimate, we observe that for $\tau>\tau_{0}$,
\begin{equation}
\begin{split}\label{lemma3best}
\int_{\tau_{0}}^{\tau}|h_{2}(\tau,x,t)|dt&<\int_{\tau_{0}}^\tau\left|\frac{\ddot{a}(t)}{\dot{a}^2(t)}\right|\frac{a^2(\tau)}{\sqrt{a^2(\tau)-a^2(t)}\sqrt{a^2(\tau)-a^2(x)}}\,dt\\
&=\frac{a(\tau)}{\sqrt{a^2(\tau)-a^2(x)}} \int_{\tau_{0}}^\tau\left|\frac{\ddot{a}(t)}{\dot{a}^2(t)}\right|\frac{a(\tau)}{\sqrt{a^2(\tau)-a^2(t)}}\,dt,
\end{split}
\end{equation}
and it follows from Lemma \ref{continuitylemma}(b) that the right side can be made arbitrarily small for $\tau$ is sufficiently close to (but greater than) $\tau_{0}$, and this may be done uniformly in $x$ for all $x$ sufficiently small.\\

\noindent Now in analogy with Eq.\eqref{gttcont}, for $\tau>\tau_{0}$,
\begin{multline}\label{dtaugttcont}
|F_{2}(\tau,x)-F_{2}(\tau_{0},0)|\leq\int_{x}^{\tau_{0}}|h_{2}(\tau,x,t)-h_{2}(\tau_{0},0,t)|dt \\+\left|\int_{\tau_{0}}^{\tau}h_{2}(\tau,x,t)dt\right|
+\left|\int_{0}^{x}h_{2}(\tau_{0},0,t)dt\right|.
\end{multline}
The third term can be made arbitrarily small by choosing $x$ sufficiently close to zero because $h_{2}(\tau_{0},0,t)$ is integrable. The second term on the right is small for $\tau$ sufficiently close to $\tau_{0}$, uniformly in $x$ by Eq.\eqref{lemma3best} and Lemma \ref{continuitylemma}(b) with $\ell(t)=\ddot{a}(t)/\dot{a}(t)^{2}$. The first term on the right side of Eq.\eqref{gttcont} can be made arbitrarily small by choosing $\tau$ sufficiently close to $\tau_{0}$ and $x$ sufficiently close to $0$ by using the fact that $h_{2}$ is continuous jointly in all its variables and that $h_{2}(\tau_{0},0,t)$ and $h_{2}(\tau,x,t)$ are both integrable on $(0,\tau_{0})$.  The case $\tau<\tau_{0}$ is similar. Thus, $\partial_{\tau}f(\tau,t_0)$ is continuous on $D$.\\

\noindent It remains to show that $\partial_{t_0}f(\tau,t_0)\partial_\tau t_0(\tau,t_0)$ is continuous at all points of the form $(\tau_0,\rho_{\mathcal{M}_{\tau_0}})$.  From Eqs. \eqref{dt0g} and \eqref{taudp},
\begin{multline}
\pder{f}{t_{0}}(\tau,t_{0})\pder{t_0}{\tau}(\tau,t_{0})\\
=\frac{\dot{a}(|t_{0}|)}{\sqrt{a^2(\tau)-a^2(t_{0})}}\int_{|t_{0}|}^\tau\frac{\ddot{a}(t)}{\dot{a}^2(t)}\frac{1}{\sqrt{a^2(\tau)-a^2(t)}}dt\left[\der{\rho_{\mathcal{M}_\tau}}{\tau}(\tau)-\pder{G}{\tau}(\tau,t_{0})\right],
\end{multline}
where 
\begin{equation}\label{dgdtau}
\pder{G}{\tau}(\tau,t_0)=-a(\tau)\dot{a}(\tau)\int_0^{t_0}\frac{a(t)}{(a^2(\tau)-a^2(t))^{3/2}}dt,
\end{equation}
and
\begin{equation}
\der{\rho_{\mathcal{M}_\tau}}{\tau}(\tau)=\frac{\dot{a}(\tau)}{a(\tau)}\int_0^\tau\left(1-\frac{a(t)\ddot{a}(t)}{\dot{a}^2(t)}\right)\frac{a(t)dt}{\sqrt{a^2(\tau)-a^2(t)}}.
\end{equation}
As above, it is sufficient to show that the function $F_{3}(\tau,x)\equiv\partial_{t_{0}}f(\tau,x)\partial_\tau t_{0}(\tau,x)$ defined on the set $\{(\tau,x): \tau\geq x\geq0\}$ is continuous at any point of the form $(\tau,x)=(\tau_{0},0)$.  Using Eq.\eqref{F0}, we may write,
\begin{equation}
F_{3}(\tau,x)= \frac{F_{0}(\tau,x)}{\sqrt{a^2(\tau)-a^2(x)}}\left[\der{\rho_{\mathcal{M}_\tau}}{\tau}(\tau)-\pder{G}{\tau}(\tau,x)\right].
\end{equation}
The function $F_{0}(\tau,x)$ was shown to be continuous at $(\tau_{0},0)$ in the proof of Theorem \ref{rhoderivg} and it is easily shown that the functions in the square brackets are continuous as well.  It follows that  $\partial_{t_0}f(\tau,t_0)\partial_\tau t_0(\tau,t_0)$ is continuous at $(\tau_0,\rho_{\mathcal{M}_{\tau_0}})$. 
\end{proof}

\noindent \textbf{Acknowledgments.} \\

\noindent J. Reschke was partially supported during the course of this research by the Leslie and Terry Cutler scholarship, and the California State University, Northridge Association of Retired Faculty and College of Science and Mathematics.

\end{document}